\documentclass[twoside,11pt]{article}

\usepackage{jmlr2e}
\usepackage{lastpage}

\jmlrheading{25}{2024}{1-\pageref{LastPage}}{2/23; Revised
10/24}{10/24}{23-0139}{Natesh S. Pillai}

\ShortHeadings{Optimal Scaling for Proximal MALA}{Natesh S. Pillai}

\title{Optimal Scaling for the Proximal Langevin Algorithm in
High Dimensions}

\firstpageno{1}

\usepackage{amsmath,amsfonts,amssymb}
\usepackage{mhequ}
\hypersetup{hidelinks}

\newcommand{\prox}{\mathrm{Prox}}
\newcommand{\dist}{\overset{\mathcal{D}}{\sim}}

\newcommand{\tr}{\mathrm{Tr}}

\newcommand{\C}{\mathcal{C}}

\newcommand{\OO}{\mathcal{O}}
\newcommand{\RR}{\mathbb{R}}
\newcommand{\EE}{\mathbb{E}}

\newcommand{\eps}{\varepsilon}
\newcommand{\pphi}{\hat{\varphi}}

\newcommand{\bra}[1]{\langle #1 \rangle}
\newcommand{\Normal}{\textrm{N}}
\newcommand{\err}{\textbf{e}}
\newcommand{\h}{\mathcal{H}}
\newcommand{\eqdef}{\overset{{\mbox{\tiny  def}}}{=}}

\newcommand{\longweak}{\Longrightarrow}
\newcommand{\opnm}[3]{\|#1\|_{{\cal L}({\cal H}^{#2},{\cal H}^{#3})}}
\newcommand{\sqrtd}[1]{{\textstyle \sqrt{#1}}}

\renewcommand\phi{\varphi}
\newcommand\rr{\mathit{R}}
\newcommand{\xmala}{x_{\mathrm{MALA}}}
\newcommand{\xprox}{x_\mathrm{Prox-MALA}}
\newcommand{\argmin}{\mathrm{argmin}}

\newtheorem{assumptions}[theorem]{Assumptions}

\newcommand{\be}{\begin{equs}}
\newcommand{\ee}{\end{equs}}

\begin{document}

\title{Optimal Scaling for the Proximal Langevin 
Algorithm in High Dimensions}

\author{\name Natesh S. Pillai \email \textsc{pillai@fas.harvard.edu} \\
       \addr Department of Statistics\\
       Harvard University\\
       MA 02138, USA}
       
\editor{Matthew Hoffman}

\maketitle

\begin{abstract}
The Metropolis-adjusted
Langevin (MALA) algorithm is a sampling algorithm that  incorporates the gradient
of the logarithm of the target density in its proposal distribution. 
In an earlier joint work \citet{pill:stu:12}, the author
had extended the seminal work of \cite{Robe:Rose:98} and showed that 
 in stationarity, MALA
applied to an $N-$dimensional approximation of
the target will take ${\cal O}(N^{\frac13})$ steps
to explore its target measure.
It was also shown in \citet{Robe:Rose:98} and \citet{pill:stu:12} that, as a consequence of the diffusion limit, the MALA algorithm is optimized at an average acceptance
probability of $0.574$. 
In \citet{pere:16}, the author introduced the proximal MALA algorithm
where the gradient of the log target density is replaced by the proximal function (mainly aimed at implementing MALA for non-differentiable target densities).  In this paper, we show that for a wide class of twice differentiable target densities, the proximal MALA enjoys the same optimal scaling as that of MALA in high dimensions and also has an average optimal acceptance probability of $0.574$. The results of this paper thus give the following practically useful guideline: for smooth target densities where it is expensive to compute the gradient while implementing MALA, users may replace the gradient with the corresponding proximal function (that can be often computed relatively cheaply via convex optimization) \emph{without} losing any efficiency gains from optimal scaling. We show this for two class of examples. First, for the product of Gaussians, we identify the optimal scale for proximal MALA and show that it is identical to MALA.
Next, following the exact framework used in \cite{pill:stu:12}, we define a version of the proximal MALA algorithm in a Hilbert space. We show that for a certain class of twice differentiable, infinite dimensional \emph{non-product} measures commonly used in applications, the proximal MALA applied to an $N-$dimensional approximation of
the target also will take ${\cal O}(N^{\frac13})$ steps
to explore the invariant measure, with an optimal acceptance probability of $0.574$.
 This confirms some of the empirical observations made in \cite{pere:16}.
\end{abstract}

\begin{keywords}
   Markov Chain Monte Carlo, Metropolis Adjusted Langevin Algorithm, Scaling limit, Diffusion Approximation, Convex Optimization, Proximal Operators, Moreau Envelope.
\end{keywords}
\section{Introduction}
\label{sec:introduction}

The Langevin diffusion in $\RR^N$
\be \label{eqn:lang}
dX_t = \nabla \log \pi^N(X_t) dt + \sqrt{2}\,dW_t
\ee
under practically realistic regularity assumptions on the measure $\pi^N$ has $\pi^N$ as its invariant measure. The Langevin algorithm has been one of the workhorses for sampling probability measures;  it is widely used in Bayesian statistics \citep{Case:Robe:04}, data assimilation, inverse problems \citep{Stua:10} and machine learning \textit{e.g.,} \citet{wel:11, lam:21}, among other areas of data science.
 The time discretization of $X_t$ with step-size $\delta$ gives rise to the Langevin proposal:
 \begin{equs} \label{eqn:MALAalgdef}
y = x + \delta \, \nabla \log \pi^N(x) + \sqrt{2 \delta} \, Z^N, \qquad Z^N \sim \Normal(0,  \mathrm{I}_N) \;.
\end{equs}
Consider a $\pi^N-$invariant Metropolis Hastings Markov chain 
$\big\{ x^{k,N} \big\}_{k \geq 1}$ obtained by proposing $y$ from the current state $x$ 
according to the kernel $q(x,y)$ given by \eqref{eqn:MALAalgdef} and then accepted with probability 
\be \label{eqn:accrej}
\alpha(x, y) = 1 \wedge {\pi^N(y) q(y,x) \over \pi^N(x) q(x,y)}.
\ee
The proposal \eqref{eqn:MALAalgdef} coupled with the accept-reject mechanism above constitutes the Metropolis Adjusted Langevin
Algorithm (MALA) \citep{Case:Robe:04}. The proposal kernel for the simpler, Random Walk Metropolis (RWM) algorithm is derived from the following random walk:
\begin{equs} \label{eqn:RWMdef}
y = x +  \,\sqrt{\delta}\, Z^N, \qquad  Z^N \sim \Normal(0,  \mathrm{I}_N) \;.
\end{equs}

An important question regarding the computational complexity of these Markov chains is how should the parameter $\delta$ vary as a function of the dimension $N$.  A well-known heuristic for choosing $\delta$ is the  following: smaller values of $\delta$ lead to high acceptance rates but the chain moves very slowly 
and therefore may not be efficient. Larger values of $\delta$ lead to larger moves, but are rejected more often because of smaller acceptance probabilities. The ``optimal scale" for 
the proposal variance  thus strikes a balance between making large moves and still having
an $\mathcal{O}(1)$ acceptance probability as a function of the dimension $N$.

To make this heuristic precise, consider the continuous interpolant of the Markov chain $X^{k,N}$:
\begin{equs} 
z^N(t)= \Bigl(\frac{t}{\Delta t} - k\Bigr)\, x^{k+1,N} 
+ \Bigl(k+1 - \frac{t}{\Delta t}\Bigr) \,x^{k,N}, 
\qquad \text{for} \qquad k \Delta t \leq t < (k+1) \Delta t. \\ \label{eqn:MCMCe}
\end{equs}
We choose the proposal variance to satisfy $\delta = \ell\Delta t$, 
with $\Delta t = N^{-\gamma}$ setting the scale in terms of
dimension and the parameter 
$\ell$  a ``tuning'' parameter which 
is independent of the dimension $N$.
We now discuss how to choose $\gamma$ and $\ell$.  

Suppose that  $\pi^N$ is the product of $N$ probability densities $\pi$,
\be \label{eqn:prodmeasure}
\pi^N(x) \propto \prod_{i=1}^N \pi(x_i).
\ee 
For this product measure, the seminal papers \citet{Robe:etal:97} and  \citet{Robe:Rose:98} respectively  showed that, \emph{in stationarity}, the ``optimal" choice for $\gamma$ that maximizes the expected squared jumping distance is $\gamma=1$  for the RWM algorithm and $\gamma=\frac13$ for the MALA. Moreover, the projection of $z^N$ into any single 
fixed coordinate direction $x_i$ converges weakly
in $C([0,T];\RR)$ to $z$, the scalar diffusion process of the form:
\begin{equs}
\frac{dz}{dt} &= h(\ell) [\log \pi(z)]'+\sqrt{2 h(\ell) }\frac{dW}{dt}\label{eq:sde}.
\end{equs}
Here $h(\ell)>0$ is a constant determined by the parameter
$\ell$ from the proposal variance. The quantity $h(\ell)$ has the interpretation as the ``speed measure'' of the limiting diffusion; see \cite{Robe:Rose:01}. Choosing $\ell$ to maximize $h(\ell)$, thus maximizing the speed of the limiting diffusion, then yields an optimal average acceptance probability
of $0.234$ for the Random Walk Metropolis Algorithm and $0.574$ for MALA. A remarkable feature of these results is that the optimal acceptance probabilities for these two algorithms are ``universal" -- they hold for a wide range of $\pi$.

The  above analysis shows that the number of steps required 
to sample the target measure
grows as $\OO(N)$ for RWM, but only as $\OO(N^{\frac13})$ for MALA.
This quantifies the efficiency gained by
use of MALA over RWM, and in particular from employing local moves
informed by the gradient of the logarithm of the target density. These theoretical analyses have inspired much further research as they give useful guidelines for implementation of MALA in high dimensions:  in addition to employing an explicit scale in the proposal variance as predicted by the theory,  one should 
``tune" the proposal variance of the RWM and MALA algorithms 
so as to have acceptance probabilities of $0.234$ and $0.574$
respectively.

\subsection{Proximal MALA algorithm} \label{sec:introproxmala}
The proximal MALA algorithm was introduced in \citet{pere:16}.
For a convex function $f:\RR^N \mapsto \RR$, $\lambda >0$ and $\| \cdot \|$ denoting the Euclidean norm, define the proximity operator \citep[also called the $\lambda$-Moreau envelope; see][]{baus:comb:11}:
\be
\prox_f^\lambda(x) = \argmin_{y \in \RR^N} \Big(f(y) + \frac{1}{2\lambda}\|y - x\|^2 \Big).
\ee
The following two extreme limits are well known for proximal functions \citep[see][chap.~12]{baus:comb:11}:
$$ \lim_{\lambda \rightarrow 0} \prox_f^\lambda(x)  =x, \quad \quad \lim_{\lambda \rightarrow \infty} f(\prox_f^\lambda(x)) = \inf_{y \in \RR^N} f(y).$$
Let $\pi^N$ be a probability density in $\RR^N$ and consider its $\lambda-$Moreau approximation \citep[see Equation (3) of][]{pere:16}:
\be
\pi^N_\lambda (x) \propto \sup_{u \in \RR^N} \pi(u) \exp\Big(-\frac{1}{2\lambda} \|u - x \|^2\Big).
\ee
If $\pi^N(x) \propto \exp(-\Psi(x))$ for a convex function $\Psi$, we have the identity:
\be \label{eqn:Morenv}
\pi^N_\lambda (x) \propto \exp \Big\{-\Psi\Big(\prox_{\Psi}^\lambda(x)\Big)\Big\}\exp\Big \{ -\frac{1}{2 \lambda} \|\prox_{\Psi}^\lambda(x) - x \|^2\Big\}.
\ee
In addition, if $\Psi$ is differentiable, we also have the identity \citep[Equation (12.28)]{baus:comb:11}:
\be \label{eqn:proxgradid0}
\frac{1}{\lambda}(x - \prox_{\Psi}^\lambda(x)) = \nabla \Psi(\prox_{\Psi}^\lambda(x)).
\ee
Equation \eqref{eqn:proxgradid0} can be thought of as an implicit gradient. Indeed, the usual explicit Euler discretization for MALA yields:
\be \label{eqn:expl}
x^{k+1,N} = x^{k, N} - \lambda \nabla \Psi(x^{k,N})  + \sqrt{2 \lambda} \, Z^N, \quad Z^N\sim \Normal(0,  \mathrm{I}_N) 
\ee 
whereas \eqref{eqn:proxgradid0} leads to the implicit update equation
\be \label{eqn:impl}
x^{k+1,N} = x^{k, N} - \lambda \nabla \Psi(x^{k+1,N}) + \sqrt{2 \lambda} \, Z^N
\ee 
or equivalently
\be \label{eqn:implicit}
x^{k+1,N}  = \prox_\Psi^\lambda(x^{k,N}) + \sqrt{2 \lambda} \, Z^N.
\ee
Motivated by \eqref{eqn:proxgradid0} and \eqref{eqn:implicit}, in  \cite{pere:16}, the author introduced the following modification of the discrete Langevin proposal \footnote{For notational consistency, we have set $2\delta = \delta'$ where $\delta'$ is the analogous parameter in Pereyra's definition; see Equation (9) of \cite{pere:16}} \eqref{eqn:MALAalgdef}:
\be \label{eqn:proxmala0}
y =  \Big(1- \frac{\delta}{\lambda}\Big) \,x +  {\delta \over \lambda} \prox_\Psi^{\lambda}(x) + \sqrt{2 \delta} \, Z^N, \qquad Z^N \sim \Normal(0,  \mathrm{I}_N) \;.
\ee
The proximal MALA Markov chain then proceeds via the accept-reject mechanism \eqref{eqn:accrej} using the proposal given in \eqref{eqn:proxmala0}. 

 In \cite{pere:16}, the author chose $\delta = \lambda$ on grounds of the stability of the resulting algorithm. We also make this choice. Thus our proximal MALA proposal is given by: 
\be \label{eqn:proxmala}
y =   \prox_\Psi^{\delta}(x) + \sqrt{2 \delta} \, Z^N, \qquad Z^N \sim \Normal(0,  \mathrm{I}_N) \;.
\ee
During the revision stages of this paper, the preprint \cite{cruc:23} was posted that significantly generalized our results. In  \cite{cruc:23}, the authors show that  $\lambda \neq \delta$ leads to sub-optimal results; see Section \ref{sec:close} for more discussion. 

One of the main reasons why the proximal MALA was introduced in \cite{pere:16} is that the proposal \eqref{eqn:proxmala} can be applied to targets even when $\Psi$ is not differentiable:
$\mathit{e.g.}$, the Laplace density $\Psi(x) = |x|$. 
 Quoting \cite{pere:16}: ``finally, similarly to other MH algorithms based on local proposals, proximal MALA may be geometrically ergodic yet perform poorly if the proposal variance $\delta$ is either too small or very large. Theoretical and experimental studies of MALA show that for many high-dimensional target densities the value of $\delta$ should be set to achieve an acceptance rate of approximately $40\%-70\%$ (Pillai et al. 2012)." 

\subsection{Motivation} \label{sec:intuit} In this paper, we show that both the MALA algorithm and the proximal-MALA algorithm \emph{enjoy} the same optimal scaling and hence the optimal acceptance probability  for a wide range of \emph{differentiable} target measures. Our results thus provide the first theoretical confirmation of the empirical observation above made in \cite{pere:16}. 
It is natural to ask why one should consider differentiable target densities for studying the performance of the proximal MALA algorithm since it was developed mainly for addressing the non-differentiable case. 
We mention a few reasons that illustrate why such a study is useful.
\begin{enumerate}
\item The proposal for the MALA algorithm is obtained from the explicit (forward) Euler discretization in \eqref{eqn:expl}:
\be
x^{k+1,N} = x^{k, N} - \lambda \nabla \Psi(x^{k,N}) + \sqrt{2\lambda} \,Z^N,
\ee
whereas the proximal MALA proposal is obtained from the implicit (backward) Euler discretization as described in \eqref{eqn:impl}. Thus is interesting to know, and far from obvious \emph{apriori}, that if this small change in the proposal obtained by the implicit discretization (proximal MALA) has better or worse scaling properties than the explicit method (MALA).  As mentioned before, one of our main contribution in this paper is to show that for a wide class of differentiable targets,  both of these methods have the same optimal scaling. In Section \ref{sec:intuit} we give a heuristic argument showing why this is the case. It is interesting to note that  even if the target distribution is non-differentiable only on a set of measure zero (\textit{e.g.}, the Laplace density, $\Psi(x) = |x|$), the proximal MALA does not achieve the $N^{-\frac{1}{3}}$ scaling as it does for smooth targets;  see \cite{cruc:23}. 

\item 
Even if the target density is differentiable, in many practical applications it may be very expensive to compute the gradient, whereas it is often cheap to compute the proximal function via convex optimization. For example, in many applied models encountered in data assimilation and Bayesian inverse problems \citep{Stua:10}, the target density is of the form: 
\be
\pi^N(\Theta|Y) \propto \exp\Big(-{1 \over 2 \sigma^2} \|Y - G(\Theta)\|^2 + h(\Theta)\Big)
\ee
where $G:\mathbb{R}^N \mapsto \mathbb{R}$ is an expensive, non-linear function to compute (such as the solution of a climate model obtained via solving a partial differential equation), $\Theta$ is a parameter we wish to compute posterior inference for, $Y$ is the observed data and $\exp(h(\Theta))$ denotes the prior distribution for $\Theta$. There have been quite a few papers recently where a sophisticated neural network was used to approximate $G$ when it is a solution of a partial differential equation (\cite{kova:23, jian:23}). In such examples, even for lower dimensional $\Theta$, it can be even more expensive to compute the gradient of a neural network so as to compute the derivative of $G$ with respect to $\Theta$. Thus there is a natural need for developing derivative free sampling algorithms that enjoy the same optimality properties of Langevin algorithms.\footnote{One such class of algorithms is the recently studied zeroth-order discretization of Langevin algorithms in \cite{roy:22}. It would be of interest to compare the performance of proximal MALA to the algorithms developed in \cite{roy:22}.}

\item Optimal scaling is not the only facet of algorithm design; many other factors must be taken into consideration. Even though our results show that the optimal scaling and the optimal acceptance probability for MALA and proximal MALA algorithms are the same, there are many examples in which these two algorithms show vastly different behavior both during the transient phase and at stationarity. It is well known that in many ODEs and PDEs, the implicit discretization is numerically more stable; see \cite{elli:93} for a construction of an implicit method that is much more stable than its explicit counterpart. 
Let us give another example from Bayesian statistics. Consider a Poisson regression model:
\be
Y|X \sim \mathrm{Poisson}(e^X), 
\ee
with the prior distribution $\pi(X) \propto \exp(-\frac{1}{2} X^2)$. The goal is to infer the posterior distribution $\pi(X|Y)$. Suppose that we observed $Y = 1$. Then 
$\pi(X|Y=1) \propto e^{-\frac{1}{2}X^2 + X - e^{X}}$. 
Since $\pi(X|Y = 1)$ has very light tails, the gradient of $\log \pi(X|Y=1)$ takes very large negative values for $X \gg 1$. 
\begin{figure}[ht]
\centering
\includegraphics[width=\textwidth]{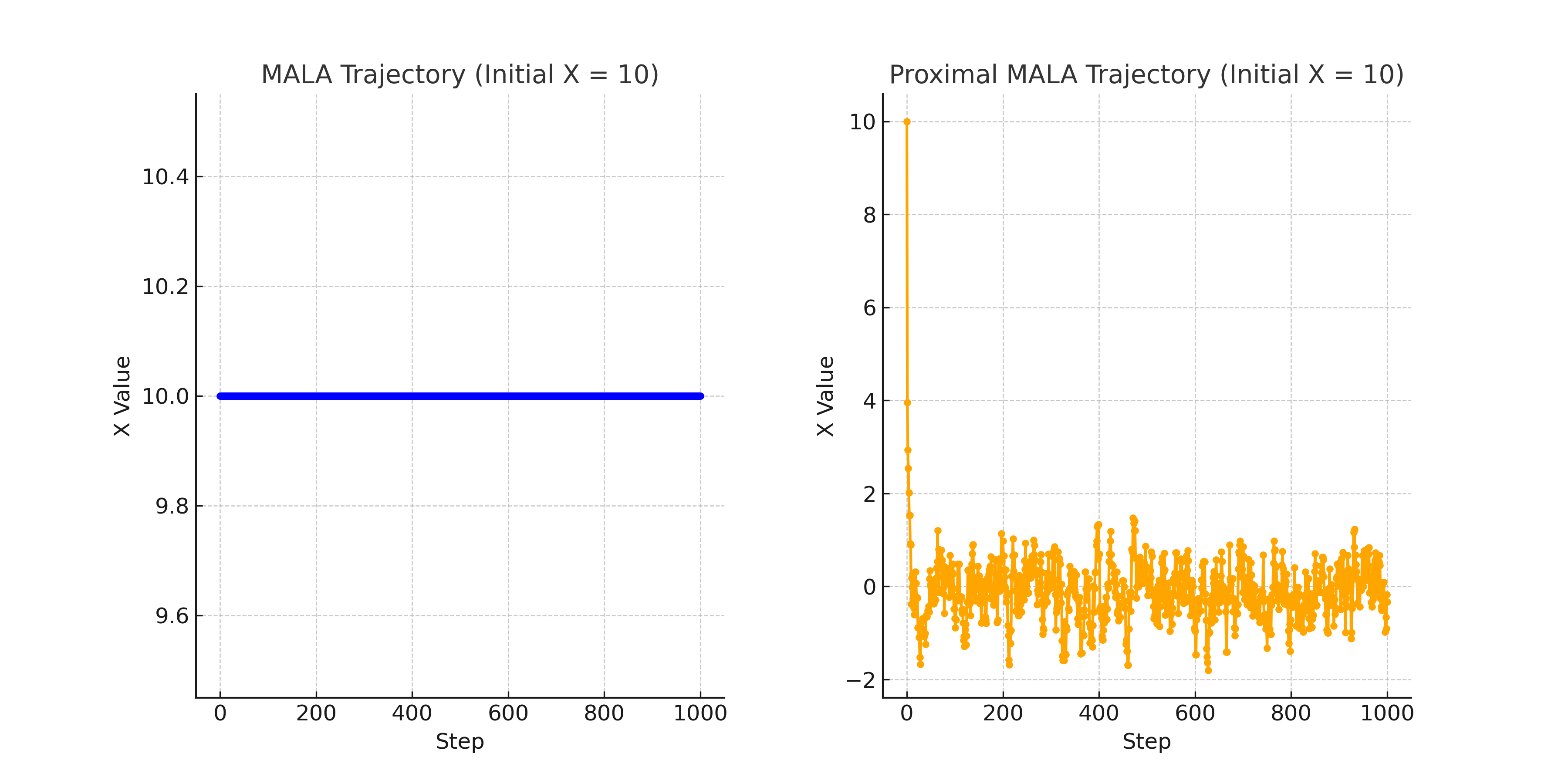}
\caption{Trajectories of the MALA and Proximal MALA algorithms for the Poisson regression example. Both chains were initialized at \( X = 10 \).}
\label{fig:mala_traj}
\end{figure}
Thus if initialized at large values of $X$, MALA and proximal MALA show very different behavior. As Figure \ref{fig:mala_traj} shows, the MALA algorithm gets stuck when initialized at $X = 10$ whereas the proximal MALA mixes after an initial burn in of $\sim 50$ steps. We conjecture that, for this example, MALA is not even geometrically ergodic whereas proximal MALA is. However, we believe that using the methods in \cite{cruc:23}, one can show that both MALA and proximal MALA have the same optimal scaling for this example.
\end{enumerate}
\subsection{Main Results}
We study the optimal scaling of proximal MALA in two contexts:
\begin{enumerate}
\item When the target measure is a product of standard Gaussians, in Theorem \ref{thm:pmalgaus} we show that the optimal scale and optimal acceptance probability for the proximal MALA algorithm is identical to that of MALA.  The recent work \cite{cruc:23} extends our results in this case to a much wider class of densities.
\item For a class of infinite dimensional non-product measures studied in \cite{pill:stu:12},  we show that the optimal scaling of $N^{-1/3}$ for MALA as worked out in \cite{Robe:Rose:98, pill:stu:12} is also optimal for the proximal MALA algorithm when the log density is convex and differentiable; see Theorem \ref{thm:main} for the formal statement of our main result.
\end{enumerate} 
The results of our paper thus give the following practically useful guideline: for smooth target densities where the gradient is expensive to compute or numerically unstable while implementing MALA, users may replace the gradient with the corresponding proximal function  \emph{without} losing any efficiency gains from optimal scaling; furthermore, users can set the proximal parameter $\delta$ to $N^{-1/3}$ and tune the algorithm to have an acceptance probability of $0.574$ just as in MALA. Of course, as discussed in Section \ref{sec:intuit}, optimal scaling alone does not give a complete picture for preferring one algorithm to the other.  Our paper takes the first theoretical steps using convex optimization to study optimal scaling of MCMC algorithms. 
\subsection{High-level explanation behind the optimal scaling} \label{sec:intuit}
Let us give a high-level explanation of why the proximal MALA enjoys the same scaling as that of MALA when $\Psi$ is differentiable. 
When $\Psi$ is smooth, it can be shown under reasonable assumptions on the second derivative of $\Psi$ that:
\be \label{eqn:intuit}
|\prox_\Psi^\delta(x) - x| = \mathcal{O}(\delta).
\ee
Consequently, setting $\lambda = \delta$ in the implicit Euler identity \eqref{eqn:proxgradid0} and using \eqref{eqn:intuit} yields that
\be 
 \prox_\Psi^\delta(x)&= x  - \delta \nabla \Psi(\prox_\Psi^\delta(x)) \\
&=  x  - \delta \nabla \Psi(x) + \rr(x,\delta), \quad \quad \rr(x,\delta) = \mathcal{O}(\delta^2). \label{eqn:mainint}
\ee
The remainder term $\rr(x,\delta)$ is $\mathcal{O}(\delta^2)$.
Comparing this with \eqref{eqn:proxmala}, we see that the proximal MALA proposal can be written as
\be \label{eqn:malaform}
y &=  x  - \delta \nabla \Psi(x) + \rr(x,\delta)  + \sqrt{2 \delta} \, Z^N, \qquad Z^N \sim \Normal(0,  \mathrm{I}_N) \\
&= \xmala + \rr(x,\delta)
\ee
where $\xmala$ is the MALA proposal.
In high dimensions, the drift term in the diffusion limit comes from $
\mathcal{O}(\delta)$ term; the $\mathcal{O}(\delta^2)$ remainder term $\rr(x,\delta)$ does not contribute to the diffusion limit and vanishes in the large $N$ limit. Our paper formalizes this observation for a class of infinite dimensional models studied in \cite{pill:stu:12}; refer to Equation \eqref{eqn:pmala-mala}, Lemma \ref{lem.drift.approx} and the related discussion in Section \ref{sec:algorithm}.

\begin{remark}
Another important theoretical aspect is to study the mixing times of proximal  MALA algorithms and obtaining non-asymptotic guarantees. See \cite{Durm:18} and the references therein. As in the original scaling papers \cite{Robe:etal:97} and  \cite{Robe:Rose:98}, we also do not study the mixing times of the proximal Markov chains in this paper. 
\end{remark}
\subsection{Infinite Dimensional Diffusions} 
Motivated by applications in data assimilation, inverse problems and Bayesian nonparametrics (see \cite{Stua:10} and  \cite{Hair:Stua:Voss:10}), the papers \cite{Matt:Pill:Stu:11} and  \cite{pill:stu:12} extended the results of product measures \cite{Robe:Rose:98} to certain infinite dimensional \emph{non-product} target measures. 
In both of these papers, the target measure of interest,
$\pi$, is on an infinite dimensional real separable Hilbert 
space $\h$ and 
is absolutely continuous with respect to a 
Gaussian measure $\pi_0$ on $\h$
with  mean zero and covariance operator $\C$. 
This framework for the analysis of MCMC in high dimensions
was first studied in the papers \cite{Besk:etal:08,BRS09,Besk:Stua:07}. 
The Radon-Nikodym derivative defining the target measure is
assumed to have the form 
\begin{equs}
\frac{d\pi}{d\pi_0}(x) = M_{\Psi} \exp ( -\Psi(x) )
\label{eqn:targmeas}
\end{equs}
for a real-valued functional 
$\Psi: \h^s \mapsto \mathbb{R}$ 
defined on a subspace $\h^s \subset \h$ that 
contains the support of the reference measure $\pi_0$; 
here $M_{\Psi}$ is a normalizing constant. 

It is proved in
\cite{Dapr:Zaby:92,Hair:etal:05,Hair:Stua:Voss:07} that
the measure $\pi$ is invariant for ${\cal H}-$valued SDEs 
(or stochastic PDEs -- SPDEs) with the form
\begin{equs} \label{eqn:spde}
\frac{dz}{dt} &= - h(\ell) \big(z + \C \nabla \Psi(z) \big)+ \sqrt{2\,h(\ell)}\,\frac{dW}{dt},
\quad z(0)=z^0 
\end{equs}
where $W$ is a  Brownian motion (see \cite{Dapr:Zaby:92}) in 
$\h$ with covariance operator $\C$ and any constant $h(\ell) > 0$.

In \cite{pill:stu:12}, the MALA algorithm was studied
when applied to a sequence of finite dimensional approximations 
of $\pi$ as in \eqref{eqn:targmeas}. 
The continuous time interpolant
of the Markov chain $z^N$ given by \eqref{eqn:MCMCe} is
shown to converge
weakly to $z$ solving \eqref{eqn:spde} in $C([0,T];\h^s)$.
Furthermore, the scaling
of the proposal variance which achieves this scaling
limit is inversely proportional to $N^{1/3}$ ($i.e.,$ corresponds to the exponent 
$\gamma= 1/3$) and the speed
of the limiting diffusion process is maximized at the
same universal acceptance probability of $0.574$ that 
was found for product measures  \cite{Robe:Rose:98}.

\subsection{Notation}

Throughout the paper we use the
following notation in order to compare 
sequences and to denote conditional expectations.
\begin{itemize}
\item Two sequences $\{\alpha_n\}$ and $\{\beta_n\}$ satisfy $\alpha_n \lesssim \beta_n$ 
if there exists a constant $K>0$ satisfying $\alpha_n \leq K \beta_n$ for all $n \geq 0$.
The notations $\alpha_n \asymp \beta_n$ means that $\alpha_n \lesssim \beta_n$ and $\beta_n \lesssim \alpha_n$.
\item
Two sequences of real functions $\{f_n\}$ and $\{g_n\}$ defined on the same set $D$
satisfy $f_n \lesssim g_n$ if there exists a constant $K>0$ satisfying $f_n(x) \leq K g_n(x)$ for all $n \geq 0$
and all $x \in D$.
The notations $f_n \asymp g_n$ means that $f_n \lesssim g_n$ and $g_n \lesssim f_n$.
\item
The notation $\EE_x \big[ f(x,\xi) \big]$ denotes
expectation with respect to $\xi$ with 
the variable $x$ fixed. 
\end{itemize}

\section{A Simple Example: Product of Gaussians}
We start with a simple case, where the target measure is the product of standard Gaussians: 
\be \label{eqn:pi1d}
 \pi^N(x) \propto \prod_{i=1}^N \exp(-x^2_i/2).
\ee
The MALA proposal for $\pi^N$ given in \eqref{eqn:pi1d} is:
\be
 y = x(1- \delta) + \sqrt{2\delta}\, Z, \quad \quad Z \sim \mathrm{N}(0,\mathrm{I}_N).
\ee
The Metropolis-Hastings acceptance ratio $\alpha(x,y)$ given in \eqref{eqn:accrej} with
\be
q(x,y) = \prod_{i=1}^N \exp\Big( -{1 \over 4\delta} \big(y_i - x_i(1-\delta)\big)^2 \Big).
\ee 
The usual calculation for finding the optimal scale proceeds as follows.
Expanding the term $L_n \equiv \log \Big( {\pi^N(y) q(y,x) \over \pi^N(x) q(x,y)} \Big)$ in $\delta$ yields \footnote{We used MATHEMATICA for obtaining this expansion; also see \cite{Robe:Rose:98}.}:
\be \label{eqn:Ln}
L_n = -\frac{\delta^{3/2}}{\sqrt{2}} \sum_{i=1}^N x_i Z_i+\frac{1}{2}\delta^2 
\sum_{i=1}^N\left(x^2_i- Z_i^2\right)+\frac{\delta^{5/2}}{\sqrt{2}} \sum_{i=1}^Nx_i Z_i -\frac{\delta^3}{4} \sum_{i=1}^Nx^2_i+\mathcal{O}\left(\delta^{7/2}\right).
\ee
Since the chain is at stationarity, the first three summands in \eqref{eqn:Ln} have expectations zero:
\be
\EE^{\pi^N} \EE_x (x  Z_i) = \EE^{\pi^N} \EE_x  \left(x^2_i- Z_i^2\right) = \EE^{\pi^N} \EE_x (x_i Z_i) = 0.
\ee
Moreover, the variance of the coefficient of the $\mathcal{O}(\delta^{3/2})$ term satisfies:
\be
\mathrm{Var}_x(\sum_{i=1}^N x_i Z_i) = \sum_{i=1}^Nx^2_i.
\ee
Thus if we set $\delta  = \ell N^{-1/3}$, using the fact that $\frac{1}{N} \sum_{i=1}^N x_i^2 \rightarrow 1$ almost surely, we obtain that 
\be \label{eqn:meanvar}
L_n \Longrightarrow Z_\ell \sim \mathrm{N}(-{\ell^3 \over 4}, {\ell^3 \over 2})
\ee
and the acceptance probability:
\be
\mathbb{E}(1 \wedge e^{L_n}) \rightarrow a(\ell) \equiv \mathbb{E}(1 \wedge e^{Z_\ell}).
\ee
In particular, $L_n = \mathcal{O}(1)$ for $\delta = N^{-1/3}$, and thus the \emph{optimal scale} that makes the size of acceptance probability equal to $\mathcal{O}(1)$ corresponds to $\delta = N^{-1/3}$. The ongoing computation generalizes for quite a large class of product measures $\pi^N$ far beyond Gaussians, and forms the basis of the diffusion limit obtained in \cite{Robe:Rose:98}. 
Finally, to have the optimal acceptance probability of $0.574$ that maximizes the speed of the limiting diffusion, all one needs to verify is that the limiting Gaussian random variable $Z_\ell$ satisfies:
\be \label{eqn:Zell}
-2\mathbb{E}(Z_\ell) = \mathrm{Var}(Z_\ell).
\ee
Indeed, once we have the relation \eqref{eqn:Zell}, the limiting diffusion has the speed measure: 
$$ h(\ell) = \ell^2 \mathbb{E}(1 \wedge e^{Z_\ell}) = 2\ell^2 \Phi(-{K \over 2} \ell^3)$$
for some constant $K$ that depends on the target measure and $\Phi$ is the CDF of the standard Gaussian distribution. As shown in Theorem 2 of \cite{Robe:Rose:98}, the value of $\ell$ that maximizes $h(\ell)$ is independent of $K$ since making the transformation $u = {K \over 2} \ell^3$ yields that 
\be
\max_{\ell} h(\ell) = 2^{5/3} K^{-2/3}  \max_{u} u^{2/3} \Phi(-u)
\ee
and the maximizer $\hat{u}$ of the latter term is independent of $K$, see Theorem 2 of \cite{Robe:Rose:98}. Thus the optimal acceptance probability is also independent of $K$: it is just $\hat{a} = 2 \Phi(-\hat{u})$.

Next, we perform the same computation for the proximal MALA algorithm.
The proximal MALA proposal for $\pi^N$ given in \eqref{eqn:pi1d} is:
\be \label{eqn:proxmalagaus}
 y = {1 \over (1+ \delta)} x + \sqrt{2\delta}\, Z, \quad \quad Z \sim \mathrm{N}(0,\mathrm{I}_N)
\ee
with the corresponding $q(x,y)$:
\be
q(x,y) = \prod_{i=1}^N \exp\Big( -{1 \over 4\delta} \big(y_i - {x_i \over (1+\delta)}\big)^2 \Big).
\ee 
\begin{theorem} \label{thm:pmalgaus}
For the proximal MALA proposal given in \eqref{eqn:proxmalagaus}, the choice of $\delta = \ell N^{-1/3}$ yields an acceptance probability of $\mathcal{O}(1)$. The limiting acceptance probability $a(\ell)$ can be expressed as $a(\ell) = \EE(1 \wedge e^{\tilde{Z_\ell}})$
where $\tilde{Z}_\ell$ is a Gaussian variable satisfying \eqref{eqn:Zell}.
\end{theorem}
\begin{proof} As before, expanding $L_n \equiv \log \Big( {\pi^N(y) q(y,x) \over \pi^N(x) q(x,y)} \Big)$ in terms of $\delta$ yields:
\be 
L_n &= -\frac{3}{\sqrt{2}} \delta^{3/2}\sum_{i=1}^N x_i Z_i+\frac{3}{2}\delta^2 \sum_{i=1}^N \left(x^2_i - Z_i^2 \right) \\
&\quad \quad + \delta^{5/2}\frac{7}{\sqrt{2}}\sum_{i=1}^N x_i Z_i+\frac{1}{4}\delta^3 \sum_{i=1}^N \left(8z_i^2-17x^2_i\right)+\mathcal{O}\left(\delta^{7/2}\right).\label{eqn:newLn}
\ee
Again, using the fact that the chain is at stationarity, we see that the summands of $\delta^{3/2}, \delta^{2}$ and $\delta^{5/2}$ in the expansion \eqref{eqn:newLn} all have mean zero. Furthermore, for the choice of $\delta = \ell N^{-1/3}$, we have
$L_n \Longrightarrow \tilde{Z}_\ell$ 
with 
$\frac{9}{2} = -2\mathbb{E}(\tilde{Z}_\ell) = \mathrm{Var}(\tilde{Z}_\ell)$
satisfying \eqref{eqn:Zell}, and the proof is finished.
\end{proof}
While we do not prove a diffusion limit, the arguments laid out in Section \ref{sec:intuit} can be used to prove a diffusion limit for any single component of the piecewise interpolant of the proximal Markov chain described above. Consequently, Theorem \ref{thm:pmalgaus} yields that the optimal acceptance probability for proximal MALA algorithm is also $0.574$ in the case where the target measure is the product of Gaussians. 

\begin{remark}
While Theorem \ref{thm:pmalgaus} is only worked out for product of Gaussians,  the result and the heuristic arguments given in Section \ref{sec:intuit}
strongly suggest that the same optimal scale and acceptance probability should hold for a large class of measures obtained as products of smooth, log-concave target densities; this was rightly confirmed in \cite{cruc:23}. This is because the optimal scale and optimal acceptance probability results are ``universal"; the specifics of target distributions should not matter. In particular, the Gaussian distribution (as used in  Theorem \ref{thm:pmalgaus}) plays no special role in optimality of MALA and nor should play a role here. We focused on this case for clarity of exposition.
\end{remark}

\section{Infinite Dimensional Target Measure}
\label{sec:prior}
We keep the framework in this paper \textbf{identical} to that of \cite{pill:stu:12} so that the reader can easily compare our results to that of the MALA algorithm obtained in that paper. The structure of proof of the diffusion limit is also identical to that of \cite{pill:stu:12}. Recall that our main goal is to show that the proximal MALA proposal has the same performance as that of the infinite dimensional MALA algorithm studied in \cite{pill:stu:12}.
Thus we are not interested in reproving the results of \cite{pill:stu:12}; instead, we merely wish to highlight only those parts where adding a proximal term (instead of the gradient) in the MALA leads to an alteration of the proof of diffusion limit worked out in \cite{pill:stu:12}. 

Let $\h$ be a separable Hilbert space of real valued functions 
with scalar product denoted by $\bra{\cdot, \cdot}$
and associated norm $\|x\|^2 = \bra{x,x}$.
Consider a Gaussian probability measure $\pi_0$ 
on $\big(\h, \| \cdot \| \big)$ with covariance 
operator $\C$. The general theory of Gaussian 
measures \cite{Dapr:Zaby:92} ensures 
that the operator $\C$ is positive and trace class. 
Let $\{\phi_j,\lambda^2_j\}_{j \geq 1}$ be the eigenfunctions
and eigenvalues of the covariance operator $\C$:
\begin{equs}
\C\phi_j = \lambda^2_j \,\phi_j, \qquad \qquad j \geq 1.
\end{equs}
We assume a normalization under which 
the family $\{\phi_j\}_{j \geq 1}$ 
forms a complete orthonormal basis in 
the Hilbert space $\h$, which
we refer to us as the Karhunen-Lo\`eve basis.
Any function $x \in \h$ can be represented in this 
basis via the expansion
\begin{equs}
x&= \sum_{j=1}^{\infty} x_j \, \phi_j, \qquad x_j \eqdef \bra{x,\phi_j}.
\label{eqn:eigenexp}
\end{equs}
Throughout this paper we will often identify the function $x$ with 
its coordinates  $\{x_j\}_{j=1}^{\infty} \in \ell^2$ in this 
eigenbasis, moving freely between the two representations. 
The Karhunen-Lo\`eve expansion (see \cite{Dapr:Zaby:92}, 
section {\it White Noise expansions}), 
refers to the fact that a realization $x$ from the Gaussian measure $\pi_0$ can 
be expressed by allowing the coordinates $\{x_j\}_{j \geq 1}$ in
\eqref{eqn:eigenexp} to be independent random 
variables distributed as $x_j \sim \Normal(0,\lambda_j^2)$. Thus, in 
the coordinates $\{x_j\}_{j \geq 1}$, 
the Gaussian reference measure $\pi_0$ 
has a product structure.\\

For every $x \in \h$ we have the representation
\eqref{eqn:eigenexp}.
Using this expansion, we define Sobolev-like spaces $\h^r, r \in \RR$, with the inner-products and norms defined by
\begin{equation}\label{eqn:Sob}
\bra{x,y}_r \eqdef \sum_{j=1}^\infty j^{2r}x_jy_j,
\quad \qquad \|x\|^2_r \eqdef \sum_{j=1}^\infty j^{2r} \, x_j^{2}.
\end{equation}
Notice that $\h^0 = \h$ and 
$\h^r \subset \h \subset \h^{-r}$ for any $r >0$.  
The Hilbert-Schmidt norm $\|\cdot\|_\C$ associated to the covariance operator $\C$ 
is defined as
\begin{equs}
\|x\|^2_\C = \sum_j \lambda_j^{-2} x_j^2.
\end{equs}
For $x,y \in \h^r$, the outer product operator in $\h^r$
is the operator $x \otimes_{\h^r} y: \h^r \to \h^r$ 
defined by $(x \otimes_{\h^r} y) z \eqdef \bra{ y, z}_r \,x$ 
for every $z \in \h^r$. 
For $r \in \RR$, let  $B_r : \h \mapsto \h$ denote the operator which is
diagonal in the basis $\{\phi_j\}_{j \geq 1}$ with diagonal entries
$j^{2r}$. The operator $B_r$ satisfies 
$B_r \,\phi_j = j^{2r} \phi_j$ 
so that $B^{\frac12}_r \,\phi_j = j^r \phi_j$. 
The operator $B_r$ 
lets us alternate between the Hilbert space $\h$ and the
Sobolev  spaces $\h^r$ via the identities
$\bra{x,y}_r = \bra{ B^{\frac12}_r x,B^{\frac12}_r y }$.
Since $\|B_r^{-1/2} \phi_k\|_r = \|\phi_k\|=1$, 
we deduce that $\{B^{-1/2}_r \phi_k \}_{k \geq 0}$ forms an 
orthonormal basis for $\h^r$.
For a positive, self-adjoint operator $D : \h \mapsto \h$, 
we define its trace in $\h^r$ by 
\begin{equs}
\label{eqn:trace}
\tr_{\h^r}(D) 
\eqdef \sum_{j=1}^\infty \bra{ (B_r^{-\frac{1}{2}} \phi_j), D (B_r^{-\frac{1}{2}} \phi_j) }_r.
\end{equs}
Since $\tr_{\h^r}(D)$ does not depend on the orthonormal basis, 
the operator $D$ is said to be trace class in $\h^r$ if $\tr_{\h^r}(D) < \infty$ for
some, and hence any, orthonormal basis of $\h^r$.
Let us define the operator $\C_r \eqdef B^{1/2}_r \,\C \,B^{1/2}_r$. 
Notice that $ \tr_{\h^r}(\C_r)=\sum_{j=1}^\infty \lambda^2_j\,j^{2r}$.
In \cite{pill:stu:12} it is shown that under the condition
\begin{equs}
\label{eq.finite.trace}
\tr_{\h^r}(\C_r) < \infty, 
\end{equs}
the support of $\pi_0$ is included in $\h^r$ in the sense that 
$\pi_0$-almost every function $x \in \h$ belongs to $\h^r$.
Furthermore, the induced distribution of $\pi_0$ on $\h^r$ is 
identical to that of a centered Gaussian measure 
on $\h^r$ with covariance operator $\C_r$. 
For example, if $\xi \dist \pi_0$, 
then $\EE\big[ \bra{\xi,u}_r \bra{\xi,v}_r \big] = \bra{u,\C_r v}_r$ 
for any functions $u,v \in \h^r$.
Thus in what follows, we alternate between 
the Gaussian measures $\Normal(0,\C)$ on $\h$ 
and  $\Normal(0,\C_r)$ on $\h^r$, for those $r$ for
which \eqref{eq.finite.trace} holds.

\subsection{Change of Measure}
\label{ssec:com}

Our goal is to sample from a measure $\pi$  
defined through the change of probability formula 
\eqref{eqn:targmeas}. As described above, 
the condition $\tr_{\h^r}(\C_r) < \infty$ implies that 
the measure $\pi_0$ has full support on $\h^r$, 
\textit{i.e.}, $\pi_0(\h^r)=1$. Consequently, 
if $\tr_{\h^r}(\C_r) < \infty$, the functional 
$\Psi(\cdot)$ needs only to be defined on $\h^r$ 
in order for the change of probability formula 
\eqref{eqn:targmeas} to be valid. In this section 
we give assumptions on the decay of the
eigenvalues of the covariance operator $\C$ of $\pi_0$ 
that ensure the existence of a real 
number $s>0$ such that $\pi_0$ has full support 
on $\h^s$. The functional $\Psi(\cdot)$ is assumed 
to be defined on $\h^s$ and we impose regularity 
assumptions on $\Psi(\cdot)$ that ensure that the 
probability distribution $\pi$ is not too different 
from $\pi_0$, when projected into directions associated
with $\varphi_j$ for $j$ large.
For each $x \in \h^s$ the derivative $\nabla \Psi(x)$
is an element of the dual $(\h^s)^*$ of $\h^s$ 
comprising linear functionals on $\h^s$.
However, we may identify $(\h^s)^*$  with $\h^{-s}$ 
and view $\nabla \Psi(x)$
as an element of $\h^{-s}$ for each $x \in \h^s$. With this identification,
the following identity holds
\begin{equs}
\| \nabla \Psi(x)\|_{\mathcal{L}(\h^s,\RR)} = \| \nabla \Psi(x) \|_{-s} 
\end{equs}
and the second derivative $\partial^2 \Psi(x)$ can 
be identified as an element 
of $\mathcal{L}(\h^s, \h^{-s})$.
To avoid technicalities we assume that $\Psi(\cdot)$ is quadratically bounded, 
with first derivative linearly bounded and second derivative globally 
bounded. Weaker assumptions could be dealt with by use of stopping time 
arguments. 
\begin{assumptions} \label{ass:1}
The covariance operator $\C$ and functional $\Psi$ satisfy the following:
\begin{enumerate}
\item {\textbf{Decay of Eigenvalues $\lambda_j^2$ of $\C$:}} 
there is an exponent $\kappa > \frac{1}{2}$ such that
\begin{equs} 
\label{eq.eigenvalue.decay}
\lambda_j \asymp j^{-\kappa}.
\end{equs}
\item {\textbf{Assumptions on} $\Psi$:} 
The function $\Psi$ is convex.
There exist constants $M_i \in \RR_+, i \leq 4$ and $s \in 
[0, \kappa - 1/2)$ such that for all $x \in \h^s$ the functional 
$\Psi:\h^s \to \RR$ satisfies
\begin{equs}
M_1 \leq \Psi(x) &\leq  M_2 \, \Big(1 +  \|x\|_s^2\Big)     \label{eqn:psi1}\\
\| \nabla \Psi(x)\|_{-s} &\leq M_3 \, \Big(1 + \|x\|_s\Big) \label{eqn:psi2}\\
\opnm{\partial^2 \Psi(x)}{s}{-s} & \leq M_4. \label{eqn:psi3}
\end{equs}
\end{enumerate}
\end{assumptions}

\noindent
\begin{remark}
The convexity of $\Psi$ is not assumed in \cite{pill:stu:12}. It is not required for the MALA algorithm. In this paper we assume the convexity of $\Psi$ so as to get a unique value for the proximal operator. This assumption is not strictly necessary for our methods to go through. However, since our key aim is to formalize the observation made in \eqref{eqn:malaform}, we avoid additional complications. 
\end{remark}
\begin{remark}
The condition $\kappa >\frac{1}{2}$ ensures that the covariance operator $\C$ 
is trace class in $\h$.
In fact, Equation \eqref{eq.finite.trace} shows that $\C_r$ 
is trace-class in $\h^r$ for any $r < \kappa - \frac{1}{2}$.
It follows that $\pi_0$ has full measure in $\h^r$ 
for any $r \in [0, \kappa - 1/2)$. In particular $\pi_0$ has full support on $\h^s$.
\end{remark}
\begin{remark}
The functional $\Psi(x)  = \frac{1}{2}\|x\|_s^2$ 
satisfies Assumptions \ref{ass:1}.
It is convex, defined on $\h^s$ and its derivative at $x \in \h^s$
is given by $\nabla \Psi(x) = \sum_{j \geq 0} j^{2s} x_j \phi_j \in \h^{-s}$ with  
$\|\nabla \Psi(x)\|_{-s} = \|x\|_s$. 
The second derivative $\partial^2 \Psi(x) \in \mathcal{L}(\h^s, \h^{-s})$ 
is the linear operator that maps $u \in \h^s$ to $\sum_{j \geq 0} j^{2s} \bra{u,\phi_j} \phi_j \in \h^s$:
its norm satisfies $\| \partial^2 \Psi(x) \|_{\mathcal{L}(\h^s, \h^{-s})} = 1$ for any $x \in \h^s$.
\end{remark}
%

\subsection{Finite Dimensional Approximation}
\label{ssec:approx}
We are interested in finite dimensional approximations 
of the probability distribution $\pi$.
To this end, we introduce the vector space spanned by 
the first $N$ eigenfunctions of the covariance operator,
\begin{equs}
X^N \eqdef {\hbox {span}}\big\{ \phi_1, \phi_2, \ldots, \phi_N \big\}.
\end{equs}
Notice that $X^N \subset \h^r$ for any $r \in [0; +\infty)$. In particular, 
$X^N$ is a subspace of $\h^s$.
Next, we define $N$-dimensional approximations 
of the functional $\Psi(\cdot)$ 
and of the reference measure $\pi_0$. To this end,
we introduce the orthogonal projection on $X^N$ 
denoted by $P^N : \h^s \mapsto X^N \subset \h^s$.
The functional $\Psi(\cdot)$ is approximated by the 
functional $\Psi^N: X^N \mapsto \RR$ defined by 
\begin{equs}
\label{eq.psi.N}
\Psi^N \eqdef \Psi \circ P^N.
\end{equs}
The approximation $\pi_0^N$ of the reference measure $\pi_0$ is the Gaussian 
measure on $X^N$ given by the law of the random variable
\begin{equs}
\pi_0^N \; \dist \; \sum_{j=1}^N \lambda_j \xi_j \phi_j
\; = \; (\C^N)^{\frac12} \, \xi^N
\end{equs}
where $\xi_j$ are i.i.d standard Gaussian random variables, $\xi^N = \sum_{j=1}^N \xi_j \phi_j$
and $\C^N = P^N \circ \C \circ P^N$. Consequently we have 
$\pi_0^N = \Normal(0, \C^N)$.
Finally, one can define the approximation $\pi^N$ 
of $\pi$ by the change of probability formula
\begin{equs} \label{eq:target3}
\frac{d\pi^{N}}{d\pi^N_0}(x) = M_{\Psi^N} \exp \big(-\Psi^N(x)\big)
\end{equs}
where $M_{\Psi^N}$ is a normalization constant. Notice that 
the probability distribution $\pi^N$ is supported on $X^N$ 
and has Lebesgue density\footnote{For ease of notation we do not distinguish 
between a  measure and its density, nor do we 
distinguish between the representation of the 
measure in $X^N$ or in coordinates in $\RR^N$} 
on $X^N$ equal to
\begin{equs}\label{eqn:pitrunc}
\pi^N(x) \; \propto \;
\exp \Big( -\frac{1}{2} \|x\|^2_{\C^N} -\Psi^N(x) \Big).
\end{equs}
In formula \eqref{eqn:pitrunc},
the Hilbert-Schmidt norm $\| \cdot \|_{\C^N}$ on $X^N$ is given by the scalar product
$\bra{u,v}_{\C^N} = \bra{u, (\C^N)^{-1} v}$ for all $u,v \in X^N$.
The operator $\C^{N}$ is invertible on $X^N$ because the 
eigenvalues  of $\C$ are assumed to be strictly positive.
The quantity $\C^N \nabla \log \pi^N(x)$ is repeatedly used in the text and
in particular appears in the function $\mu^N(x)$ given by
\begin{equs} \label{eqn:mun}
\mu^N(x) = -\Big( P^N x + \C^N \nabla \Psi^N(x) \Big)
\end{equs}
which is $\C^N \nabla \log \pi^N(x).$
This function is the drift of an ergodic Langevin diffusion 
that leaves $\pi^N$ invariants. Similarly, one defines the 
function $\mu: \h^s \to \h^s$ given by
\begin{equs}
\label{eqn:mu}
\mu(x) = - \Big( x + \C \nabla \Psi(x) \Big)
\end{equs}
which is 
$\C \nabla \log \pi(x)$.
In Lemmas 4.1 and 4.3 of \cite{pill:stu:12}, it is shown that for $\pi_0$-almost every 
function $x \in \h$, we have $\lim_{N \to \infty} \mu^N(x) = \mu(x)$; see Section \ref{sec:approx} below. 
This quantifies the manner in which $\mu^N(\cdot)$ is an approximation of 
$\mu(\cdot)$.

The next lemma gathers 
various regularity estimates on the functional $\Psi(\cdot)$ 
and $\Psi^N( \cdot)$ that are repeatedly used in the sequel. 
These are simple consequences of Assumptions \ref{ass:1} and proofs 
can be found in \cite{Matt:Pill:Stu:11}  and \cite{pill:stu:12}.
\begin{lemma} \label{lem:regularity} {\bf (Properties of $\Psi$)} 
Let the functional $\Psi(\cdot)$ satisfy Assumptions \ref{ass:1} and 
consider the functional $\Psi^N(\cdot)$ defined by Equation \eqref{eq.psi.N}. 
The following estimates hold.
\begin{enumerate}
\item
The functionals $\Psi^N: \h^s \to \RR$ satisfy the same conditions
imposed on $\Psi$ given by Equations \eqref{eqn:psi1},
\eqref{eqn:psi2} and \eqref{eqn:psi3} 
with constants that can be chosen independent of $N$.
\item The function $\C \nabla \Psi: \h^s \to \h^s$ is globally Lipschitz on $\h^s$:
there exists a constant $M_5 > 0$ such that
\begin{equs}
\| \C\nabla \Psi(x) - \C\nabla \Psi(y) \|_s \leq M_5 \, \| x-y \|_s
\qquad \qquad \forall x,y \in \h^s.
\end{equs}
Moreover, the functions $\C^N \nabla \Psi^N: \h^s \to \h^s$ also satisfy 
this estimate with a constant that can be chosen independent of $N$.
\item
The functional $\Psi(\cdot): \h^s \to \RR$ satisfies a ``one-sided" 
Taylor formula\footnote{We extend $\langle \cdot,\cdot \rangle$ from an
inner-product on $\h$ to the dual pairing between
$\h^{-s}$ and $\h^s$.}. There exists a constant $M_6 > 0$ such that
\begin{equs} \label{eqn:2nd-Taylor}
\Psi(y) - \Big( \Psi(x) + \bra{\nabla \Psi(x), y-x} \Big)
\leq M_6 \, \|x - y \|_s^2 
\qquad \forall x,y \in \h^s.
\end{equs}
Moreover, the functionals $\Psi^N(\cdot)$ also satisfy 
the above estimates with a constant that can be chosen independent of $N$.
\end{enumerate}
\end{lemma}

\begin{remark}
\label{rem:ts}
The regularity Lemma \ref{lem:regularity} shows in particular that the 
function $\mu: \h^s \to \h^s$ defined by \eqref{eqn:mu} 
is globally Lipschitz on $\h^s$. 
Similarly, it follows that $\C^N \nabla \Psi^N:\h^s \to \h^s$ 
and $\mu^N:\h^s \to \h^s$ given by \eqref{eqn:mun}
are globally Lipschitz with Lipschitz constants that 
can be chosen uniformly in $N$.
\end{remark}

\section{The proximal MALA in Hilbert space}
In this section, we construct a version of the proximal MALA algorithm of \cite{pere:16} in 
the Hilbert space $\h^s$. The proximal operators are well defined in an infinite dimensional Hilbert space. The reader is referred to \cite{baus:comb:11} for a book length treatment.
For a function $g: \h^s \mapsto (-\infty, \infty]$ and $\lambda > 0$, define the proximal function
\be \label{eqn:prox}
\prox_g^\lambda(x) = \argmin_{y \in \h^s} \Big(g(y) + \frac{1}{2\lambda} \|x - y\|^2_s \Big).
\ee
If $g$ is convex, Proposition 12.15 of \cite{baus:comb:11} yields that
$\prox_g^\lambda(x)$ is convex and differentiable. Moreover the minimizer in \eqref{eqn:prox} is unique due to the convexity of $g$. 
Similar to \eqref{eqn:Morenv}, define the function $\mathrm{E}^\lambda_g$:
\be \label{eqn:Morenv1}
\mathrm{E}^\lambda_g (x) \propto \exp \Big\{-g\Big(\prox_{g}^\lambda(x)\Big)\Big\}\exp\Big \{ -\frac{1}{2 \lambda} \|\prox_{g}^\lambda(x) - x \|^2\Big\}.
\ee
The function  $-\log \mathrm{E}^\lambda_g(x)$
is the $\lambda$-Moreau-Yoshida envelope of $g$. 
We also have the identity (\cite{baus:comb:11}, Proposition 12.29 and Corollary 17.6):
\be \label{eqn:proxgradid}
-\nabla \log \mathrm{E}^\lambda_g(x) = \frac{1}{\lambda}(x - \prox_{g}^\lambda(x)) = \nabla g(\prox_{g}^\lambda(x)).
\ee

\subsection{The Proximal-MALA Algorithm}\label{sec:algorithm}
Recall from \eqref{eqn:pitrunc} that our target measure is
\begin{equs}
\pi^N(x) \; \propto \;
\exp \Big( -\frac{1}{2} \|x\|^2_{\C^N} -\Psi^N(x) \Big).
\end{equs}

Our algorithm is motivated by the fact that the probability 
measure $\pi^N$ defined by Equation \eqref{eq:target3} 
is invariant with respect to the Langevin diffusion process
\begin{equs} \label{eqn:spdeN}
\frac{dz}{dt} &= \mathcal{C}^N \nabla \log \pi^N(z) + \sqrt{2}\,\frac{dW^N}{dt} \\
&= \C^N  \mu^N(z) + \sqrt{2}\,\frac{dW^N}{dt}
\end{equs}
where $W^N$ is a  Brownian motion in 
$\h^s$ with covariance operator $\C^N$ and $\mu^N$ is as defined in \eqref{eqn:mun}. 

To obtain a proximal algorithm that is analogous to Pereyra's algorithm given in  \eqref{eqn:proxmala}, we replace $\Psi^N(x)$ by its $\delta$-Moreau-Yoshida envelope to obtain:
\be
\pi^N_\lambda(x) 
&\propto \exp \Big( -\frac{1}{2} \|x\|^2_{\C^N} \Big) \mathrm{E}^\lambda_{\Psi^N} (x)
\ee
where $\mathrm{E}^\delta_{\Psi^N} (x)$  is defined as in Equation \eqref{eqn:Morenv1} with
$g = \Psi^N$. 
Now the usual MALA proposal for $\pi^N_\lambda(x)$ gives 
the analogue to Pereyra's proximal algorithm. 
Indeed, the MALA proposal for $\pi^N_\lambda(x)$ (with the choice of $\lambda = \delta$) gives
\begin{equs}\label{eqn:pproposal0}
y &  = x + \delta \, \C^N \nabla \log \pi^N_{\lambda}(x)  +  \sqrtd{2 \delta } \, (\C^N)^{\frac{1}{2}}\xi^N \\
&  = x + \delta \, \C^N \Big(- (\C^N)^{-1} x + \nabla \log \mathrm{E}^\delta_{\Psi^N} (x)\Big)  +  \sqrtd{2 \delta } \, (\C^N)^{\frac{1}{2}}\xi^N \\
&= (1 - \delta - \C^N)x + \C^N \prox_{\Psi^N}^{\delta}(x)  + \sqrtd{2 \delta } \, (\C^N)^{\frac{1}{2}}\xi^N \qquad 
\delta = \ell N^{-\frac{1}{3}} \\
&\equiv \xprox
\end{equs}
where the third equality follows from Equations \eqref{eqn:Morenv1} and \eqref{eqn:proxgradid}.


 Applying \eqref{eqn:proxgradid} with $\Psi  = \Psi^N$ and $\lambda = \delta$, we obtain that
\be
\prox_{\Psi^N}^{\delta}(x) &= x - \delta \nabla{\Psi^N}(\prox_{\Psi^N}^{\delta}(x)) \\
&\approx x - \delta \nabla{\Psi^N}(x) + \mathcal{O}(\delta^2).
\ee
Consequently, on $X^N$,
\be
(1 - \delta - \C^N)x + \C^N \prox_{\Psi^N}^{\delta}(x) & \approx  x - \delta(P^N x + \C^N \nabla \Psi^N(x)) \\
& = x + \delta \mu^N(x). \label{eqn:langapp}
\ee
Let
\be \label{eqn:xmala}
\xmala &= x + \delta \mu^N(x)
+ \sqrtd{2 \delta} \, (\C^N)^{\frac{1}{2}}\xi^N
\qquad\text{where} \qquad 
\delta = \ell N^{-\frac{1}{3}}
\ee
denote the usual MALA proposal obtained from the Euler discretization of the infinite dimensional diffusion \eqref{eqn:spdeN}.
Notice that $(\C^N)^{\frac{1}{2}}\xi^N \dist \Normal(0, \C^N)$.
The calculation done in \eqref{eqn:langapp} shows that our proximal MALA proposal \eqref{eqn:pproposal0} closely tags the MALA proposal:
\be \label{eqn:pmala-mala}
\xprox &=  \xmala + \rr^N(x,\delta) 
\end{equs}
where the term
\be \label{eqn:remx}
\rr^N(x,\delta) \equiv \delta\, \C^N \Big(\prox_{\Psi^N}^{\delta}(x) -  x \Big)
\ee
can be thought of as the added ``error" induced by the proximal MALA proposal as compared to MALA. As shown in Lemma \ref{lem:Remest},  we have  $\|\rr^N(x,\delta)\|_{\C^N} \lesssim \delta^2 (1+ \|x\|_s) = \mathcal{O}(\delta^2)$. As in the product measure case, for optimal scaling only terms of $\mathcal{O}(\delta^{3/2})$ and lower order contribute; thus the contribution from this remainder term to the scaling drops out in the large $N$ limit. Consequently, the optimal scaling and the diffusion limits for the proximal MALA algorithm follow from the corresponding results for the MALA algorithm. 

For streamlining further calculations, we will write the $\xprox$ proposal from \eqref{eqn:pproposal0} as
\be \label{eqn:proposal}
y= x + \delta \mu^N(x) + \rr^N(x,\delta) +
 \sqrtd{2 \delta} \, (\C^N)^{\frac{1}{2}}\xi^N
\qquad\text{where} \qquad 
\delta = \ell N^{-\frac{1}{3}}.
\ee

\subsection{Time evolution of the proximal MALA chain}
\noindent We introduce a related parameter 
$$\Delta t:=\ell^{-1}\delta=N^{-\frac13}$$
which will be the natural time-step for the
limiting diffusion process derived from the proposal
above, after inclusion of an accept-reject mechanism.
The scaling of $\Delta t$, and hence $\delta,$ 
with $N$ will ensure
that the average acceptance probability is $\mathcal{O}(1)$
as $N$ grows. 

Following \cite{pill:stu:12}, we will study the Markov chain $x^N = \{x^{k,N}\}_{k \geq 0}$
resulting from Metropolizing the proximal proposal \eqref{eqn:proposal} when it is 
started at stationarity: 
the initial position $x^{0,N}$ is distributed as $\pi^N$ and
thus lies in $X^N$. Therefore, the Markov chain evolves 
in $X^N$; as a consequence, only the first $N$ components of 
an expansion in the eigenbasis of $\C$ 
are nonzero and the algorithm can be implemented in $\RR^N$. However the analysis is cleaner 
when written in $X^N \subset \h^s$. 
The acceptance probability only depends on 
the first $N$ coordinates of $x$ and $y$ and has the form
\begin{equs} \label{eqn:accprob}
\alpha^N(x,\xi^N) = 1 \wedge \frac{\pi^N(y) T^N(y,x)}{\pi^N(x) T^N(x,y)} = 1 \wedge e^{ Q^N(x,\xi^N) }
\end{equs}
where the proposal $y$ is given by Equation \eqref{eqn:proposal}.
The function $T^N(\cdot, \cdot)$ is the density of the Langevin proposals \eqref{eqn:proposal} 
and is given by
\begin{equs}
T^N(x,y) \propto \exp\Big\{ -\frac{1}{4  \delta} \| y - x -  \delta\mu^N(x)  - \rr^N(x,\delta) \|^2_{\C^N} \Big\}.
\end{equs}
The local mean acceptance probability $\alpha^N(x)$ is
defined by
\begin{equs} \label{e.loc.accept}
\alpha^N(x) = \EE_x\big[ \alpha^N(x,\xi^N) \big]. 
\end{equs}
It is the expected acceptance probability when the algorithm stands 
at $x \in \h$.
The Markov chain $x^N = \{x^{k,N}\}_{k \geq 0}$ can also be expressed as
\begin{equs} \label{eqn:MALA_algorithm}
\left\{
    \begin{array}{ll}
    y^{k,N}	&= x^{k,N} +  \delta \mu^N(x^{k,N}) + \rr^N(x^{k,N},\delta) +
    \sqrtd{2  \delta} \, (\C^N)^{\frac{1}{2}} \; \xi^{k,N} \\
    x^{k+1,N}	&= \gamma^{k,N} y^{k,N} + (1- \gamma^{k,N})\, x^{k,N} \\
    \end{array}
\right.
\end{equs}
where $\xi^{k,N}$ are i.i.d samples distributed as $\xi^N$ and 
$\gamma^{k,N} = \gamma^N(x^{k,N},\xi^{k,N})$ creates a
Bernoulli random sequence 
with $k^{th}$ success probability $\alpha^N(x^{k,N},\xi^{k,N})$.
We may view the Bernoulli random variable as 
$\gamma^{k,N} = 1_{\{U^k < \alpha^N(x^{k,N},\xi^{k,N})\}}$ where 
$U^k \dist \text{Uniform}(0,1)$ is independent from $x^{k,N}$ and $\xi^{k,N}$. \\

In summary, the Markov chain that we have described in $\h^s$ is, 
when projected onto $X^N$, equivalent to a proximal
MALA algorithm on $\RR^N$ for the Lebesgue density \eqref{eqn:pitrunc}. 
Recall that the target measure $\pi$ in \eqref{eqn:targmeas} is the
invariant measure of the SPDE \eqref{eqn:spde}.  Our goal is to obtain an 
invariance principle for the continuous interpolant \eqref{eqn:MCMCe} 
of the Markov chain $x^N = \{x^{k,N}\}_{k \geq 0}$ started in stationarity, 
$i.e$, to show weak convergence in $C([0,T]; \h^s)$ of 
$z^{N}(t)$ to the solution $z(t)$ of the SPDE \eqref{eqn:spde}, 
as the dimension $N \rightarrow \infty$.

\section{Main Result}
In this section, we present the main result of this paper. 
Consider the constant $\alpha(\ell) =  \EE\big[ 1 \wedge e^{Z_{\ell}}\big]$ where
$Z_{\ell} \dist \Normal(-\frac{\ell^3}{4}, \frac{\ell^3}{2})$ and define the speed function
\begin{equs} \label{eqn:speedfun}
h(\ell) = \ell \alpha(\ell).
\end{equs}
The quantity $\alpha(\ell)$ represents the limiting expected acceptance probability of the MALA algorithm
while $h(\ell)$ is the asymptotic speed function of the limiting diffusion. 
\begin{theorem}\label{thm:main}
Let the initial condition $x^{0,N}$ of the proximal MALA algorithm be such that $x^{0,N} \sim \pi^N$ and let $z^N(t)$ be a piecewise linear, continuous interpolant of the proximal MALA algorithm \eqref{eqn:MALA_algorithm} with $\Delta t = N^{-1/3}$. Then, for any $T > 0$, $z^N(t)$ converges weakly in $C([0,T],\h^s)$ to the diffusion process $z(t)$ given 
by
\begin{equs} \label{eqn:spdelim}
\frac{dz}{dt} &= - h(\ell) \big(z + \C \nabla \Psi(z) \big)+ \sqrt{2\,h(\ell)}\,\frac{dW}{dt},
\quad z(0)=z^0 \sim \pi
\end{equs} 
with the constant $h(\ell)$ as given in \eqref{eqn:speedfun}. Choosing $\ell$ so as to maximize the speed function $h(\ell)$ leads to the acceptance probability of 0.574 for the proximal MALA algorithm. 
\end{theorem}

\begin{remark}
The fact that choosing $\ell$ so as to maximize the speed function $h(\ell)$ leads to the optimal universal acceptance probability of 0.574 is known since 
\cite{Robe:Rose:98}, and is also shown in \cite{pill:stu:12}. Thus to prove Theorem \ref{thm:main}, we need only establish the diffusion limit.
\end{remark}

\subsection{Proof Strategy}
The acceptance probability of the proposal  \eqref{eqn:proposal} is equal 
to $\alpha^N(x,\xi^N) = 1 \wedge e^{Q^N(x,\xi^N)}$ and the quantity 
$\alpha^N(x) = \EE_x[\alpha^N(x, \xi^N)]$ given by
\eqref{e.loc.accept} represents the mean acceptance 
probability when the Markov chain $x^N$ stands at $x$.
Recall the quantity $Q^N$ in Equation \eqref{eqn:accprob}. 
This quantity may be expressed as
\begin{equs} 
Q^N(x, \xi^N) 
&= -\frac{1}{2} \Big( \|y\|_{\C^N}^2 - \|x\|_{\C^N}^2 \Big) -\Big( \Psi^N(y)-\Psi^N(x) \Big) \\
&\quad -\frac{1}{4 \delta} \Big\{ 
\| x - y - \delta \mu^N(y)  - \rr^N(y,\delta)\|^2_{\C^N} 
- \| y - x -\delta \mu^N(x) - \rr^N(x,\delta) \|^2_{\C^N} \Big\}. \\ \label{eqn:QN}
\end{equs}

The main observation (also used in \cite{pill:stu:12})
is that $Q^N(x,\xi^N)$ can be approximated by a Gaussian random variable
\begin{equs} \label{e.gauss.approx}
Q^N(x,\xi^N) \; \approx \; Z_{\ell}
\end{equs}
where $Z_{\ell} \dist \Normal(-\frac{\ell^3}{4}, \frac{\ell^3}{2})$. 
These approximations 
are made rigorous in Lemma \ref{lem:Gaussian_approx} and Lemma \ref{lem:concentration}. 
Therefore, the Bernoulli random variable $\gamma^{N}(x,\xi^N)$ with success probability 
$1 \wedge e^{Q^N(x,\xi^N)}$ can be approximated by a Bernoulli random variable, independent of $x$,
with success probability equal to
\begin{equs} \label{e.limit.acc.prob}
\alpha(\ell) = \EE\big[ 1 \wedge e^{Z_{\ell}} \big].
\end{equs}
Thus, the limiting acceptance probability of the MALA algorithm 
is as given in Equation \eqref{e.limit.acc.prob}.\\
Recall that $\Delta t=N^{-\frac13}.$ With this notation
we introduce the drift function $d^N:\h^s \to \h^s$ given by
\begin{equs} \label{eq:drift}
d^N(x) &= \big( h(\ell) \Delta t\big)^{-1} \EE\big[ x^{1,N} - x^{0,N} \,|x^{0,N} = x \big]
\end{equs}
and the martingale difference array $\{\Gamma^{k,N}: k \geq 0\}$ defined by $\Gamma^{k,N} = \Gamma^N(x^{k,N}, \xi^{k,N})$
with
\begin{equs} \label{eq:mart}
\Gamma^{k,N} &= \big( 2 h(\ell) \Delta t \big)^{-\frac12} \Big( x^{k+1,N} - x^{k,N} - h(\ell) \Delta t \; d^N(x^{k,N}) \Big).
\end{equs}
The normalization constant $h(\ell)$ defined in Equation \eqref{eqn:speedfun}
ensures that the drift function $d^N$ and the martingale difference array $\{\Gamma^{k,N}\}$
are asymptotically independent from the parameter $\ell$. 
The drift-martingale decomposition of the Markov chain $\{x^{k,N}\}_k$ then reads
\begin{equs} \label{eqn:drift-mart}
x^{k+1,N} - x^{k,N} = h(\ell) \Delta t d^N(x^{k,N}) 
+ \sqrtd{2 h(\ell) \Delta t} \; \Gamma^{k,N}.
\end{equs}
Lemma \ref{lem.drift.approx} and Lemma \ref{lem:diffus}
exploit the Gaussian behaviour of $Q^N(x,\xi^N)$ described in Equation \eqref{e.gauss.approx}
in order to give quantitative versions of the following approximations,
\begin{equs} \label{e.drift.noise.approx}
d^N(x) \;\approx\; \mu(x)\; 
\qquad \qquad \text{and} \qquad \qquad
\Gamma^{k,N} \;\approx\; \Normal(0,C)
\end{equs}
where $\mu(x) = -\Big( x + C \nabla \Psi(x)\Big)$.
From Equation \eqref{eqn:drift-mart} it follows that for large $N$
the evolution of the Markov chain ressembles the Euler
discretization of the limiting diffusion \eqref{eqn:spde}.
The next step consists of 
proving an invariance principle for a rescaled version 
of the martingale difference array $\{\Gamma^{k,N}\}$. 
The continuous process $W^N \in \C([0;T], \h^s)$
is defined as
\begin{equs} \label{e.WN}
W^N(t) = \sqrt{\Delta t} \; \sum_{j=0}^k \Gamma^{j,N} \;+\; \frac{t - k\Delta t}{\sqrt{\Delta t}} \; \Gamma^{k+1,N}
\qquad \text{for} \qquad
k \Delta t \leq t < (k+1) \Delta t.
\end{equs}
The sequence of processes $\{W^N\}$ converges weakly in 
$\C([0;T], \h^s)$ to a Brownian motion $W$ in $\h^s$ with covariance operator equal to $C_s$.
Indeed, Proposition \ref{lem:bweakconv} proves the stronger result
\begin{equs}
(x^{0,N},W^N) \Longrightarrow (z^0,W)
\end{equs}
where $\Longrightarrow$ denotes weak convergence in $\h^s \times \C([0;T],\h^s)$ and $z^0 \dist \pi$ is independent 
of the limiting Brownian motion $W$. Once we have the invariance principle and the converge of the drift and diffusion terms, the ``Master Theorem"  in \cite{pill:stu:12} (see Proposition 3.1 of \cite{pill:stu:12}) gives the required diffusion limit. 

\section{Proof of the Main Result}
In this section, we give the proof of the Theorem \ref{thm:main}. 
To this end, we use Proposition 3.1 of \cite{pill:stu:12}. According to Proposition 3.1 of \cite{pill:stu:12}, to show the diffusion limit, we must show the following three conditions.
\begin{enumerate}
\item 
{\bf Convergence of initial conditions:} $\pi^N$ 
converges in distribution to the probability measure $\pi$ where $\pi$ has a finite first moment, that is  
$\EE^{\pi}[ \|x\|_s ] < \infty$.
\item
{\bf Invariance principle:}
the sequence $(x^{0,N}, W^N)$ defined by Equation 
\eqref{e.WN} converges weakly in 
$\h^s \times \C([0,T],\h^s)$ to $(z^0, W)$ where $z^0 \dist \pi$ 
and $W$ is a Brownian motion in $\h^s$, independent from $z^0$, 
with covariance operator $C_s$.
\item
{\bf Convergence of the drift:}
there exists a globally Lipschitz function $\mu:\h^s \to \h^s$ 
that satisfies 
\begin{equs}
\lim_{N \to \infty} \; \EE^{\pi^N}\left[ \|d^N(x)-\mu(x)\|_s \right] = 0. 
\end{equs}
\end{enumerate}

Item (1.) above follows from Lemma 4.3 of \cite{pill:stu:12}); also see Section \ref{sec:approx} below. 
Item (2.) is proved in Proposition \ref{lem:bweakconv}.  
Item (3.) is proved in  Lemma \ref{lem.drift.approx}.
Thus we have established all three conditions required by Proposition 3.1 of \cite{pill:stu:12} and thus the proof of our main result is finished. 
\section{Key Estimates}
\label{sec:key.estimates}
In this section, we prove some key estimates for the proximal operator, and 
and also collect some key approximation properties of $\mu^N$ and $\pi^N$ from 
\cite{pill:stu:12}. These properties will be repeatedly used throughout.  
\subsection{Approximation properties of $\mu^N$ and $\pi^N$}\label{sec:approx}
\begin{itemize}
\item For $\pi_0$-almost every function $x \in \h^s$, 
the approximation $\mu^N(x) \approx \mu(x)$ holds as $N$ goes to infinity. Indeed, under Assumption \ref{ass:1}, the sequences of functions $\mu^N:\h^s \to \h^s$ satisfies (see Lemma 4.1 of \cite{pill:stu:12}),
\begin{equs} \label{eqn:muntomu}
\pi_0 \Big( \Big\{ x \in \h^s: \; \lim_{N \to \infty} \; \| \mu^N(x) - \mu(x) \|_s = 0 \; \Big\} \Big) = 1.
\end{equs}

\item Under the Assumptions \ref{ass:1} the normalization 
constants $M_{\Psi^N}$ are uniformly bounded so that for any measurable 
functional $f:\h \mapsto \mathbb{R}$, 
we have from Lemma 4.3 of \cite{pill:stu:12} that
\begin{equs}
\EE^{\pi^N} \big[ |f(x)| \big] \lesssim \EE^{\pi_0}\big[ |f(x)| \big]. 
\end{equs}
Moreover, the sequence of probability measure $\pi^N$ satisfies
$\pi^N  \;\longweak\; \pi$ where $\longweak$ denotes weak convergence in $\h^s$.
\item By Fernique's theorem \cite{Dapr:Zaby:92}, for any exponent $p \geq 0$ 
we have $$\EE^{\pi^0} \big[ \|x\|_s^p \big] < \infty.$$ We also have that for any $p \geq 0$
\begin{equs}
\sup_{N \in \mathbb{N}} \EE^{\pi^N}\big[ \|x\|_s^p \big] \;<\; \infty.
\end{equs}
\end{itemize}
\subsection{Estimates involving proximal functions and the remainder term}
Recall the constant $M_6$ from \eqref{eqn:2nd-Taylor}.
\begin{lemma} \label{lem:proxfunest}
For any $x \in \h^s$ and $N \in \mathbb{N}$ and for all $\delta < {1 \over 2M_6}$,
\be
\|\prox_{\Psi^N}^\delta(x) - x\|_s \lesssim  \delta(1+ \|x\|_s).
\ee
\end{lemma}

\begin{proof}
Set $x^* = \prox_{\Psi^N}^\delta(x)$. 
Since $x^*$ minimizes the map:
\be
y \mapsto \Big( \Psi^N(y) + {1 \over 2 \delta} \|y - x\|_s^2 \Big),
\ee
from our assumptions in \eqref{eqn:2nd-Taylor} and \eqref{eqn:psi2}, it follows that
\be
{1 \over 2 \delta} \|x^* - x\|_s^2 &\leq \Psi^N(x) - \Psi^N(x^*) = |\Psi^N(x^*) - \Psi^N(x)|\\
&\leq |\langle \nabla \Psi^N(x), x^*-x \rangle| +  M_6 \|x^* - x \|_s^2 \\
&\leq M_3(1 + \|x\|_s)  \| x^*-x \|_s +   M_6 \|x^* - x \|_s^2.
\ee
Dividing by the term $\|x^* - x\|_s$ throughout and simplifying yields
\be
\|x^* - x\|_s   \leq \delta {M_3 \over (1 - 2\delta M_6)} (1 + \|x\|_s)   \lesssim \delta  (1 + \|x\|_s) 
\ee
and the proof is done.
\end{proof}

\begin{lemma}\label{lem:Remest}
Recall the remainder term $\rr^N(x,\delta)$ from \eqref{eqn:remx}.
For any $x \in \h^s$, $N \in \mathbb{N}$ and  for all $\delta < {1 \over 2M_6}$,
\be
\|\rr^N(x,\delta)\|_{\C^N} \lesssim \delta^2  (1 + \|x\|_s), \quad \quad \|\rr^N(x,\delta)\|_{s} \lesssim \delta^2  (1 + \|x\|_s).
\ee
\end{lemma}

\begin{proof}
Set $x^* = \prox_{\Psi^N}^\delta(x)$. 
Then $\rr^N(x,\delta) = \delta\, \C^N (x^* -  x )$.
Thus 
\be
\|\rr^N(x,\delta)\|^2_{\C^N} &= \langle \rr^N(x,\delta), (\C^{N})^{-1} \rr^N(x,\delta) \rangle \\
&= \delta^2 \langle \C^N (x^* -  x), (x^* -  x ) \rangle \\
&\lesssim \delta^2 \|x^* - x\|_s^2 \lesssim \delta^4  (1 + \|x\|^2_s)
\ee
where the last inequality follows from Lemma \ref{lem:proxfunest} showing the first inequality. The second inequality follows similarly:
\be
\|\rr^N(x,\delta)\|^2_{s} = \delta^2 \| \C^N (x^* -  x )\|^2_s \lesssim \delta^2 \|x^* - x\|_s^2 \lesssim \delta^4  (1 + \|x\|^2_s)
\ee
 and the proof is done.
\end{proof}
\noindent
Next lemma shows that the size of the jump $y-x$ is of order $\sqrt{\Delta t}$.
\begin{lemma} \label{lem:size:prop}
Consider $y$ given by \eqref{eqn:proposal}.
Under Assumptions \ref{ass:1}, for any $p\geq 1$ 
we have 
\begin{equs}
\EE_x^{\pi^N} \big[ \|y-x\|^p_s \big] \lesssim (\Delta t)^{\frac{p}{2}} \cdot (1 + \|x\|^p_s). 
\end{equs}
\end{lemma}
\begin{proof}
Under Assumption \ref{ass:1} the function $\mu^N$ is globally Lipschitz 
on $\h^s$, with Lipschitz constant that can be chosen independent from $N$. 
Thus using Lemma \ref{lem:Remest} we obtain that 
\begin{equs}
\|y-x\|_s &\lesssim \Delta t (1+\|x\|_s) +  \|\rr^N(x,\delta)\|_s  + \sqrtd{\Delta t}\, \| \C^{\frac{1}{2}}\xi^N \|_s \\
&\lesssim \Delta t (1+\|x\|_s) + (\Delta t)^2  (1 + \|x\|_s) \,  + \sqrtd{\Delta t} \,\| \C^{\frac{1}{2}}\xi^N \|_s \\
&\lesssim \Delta t (1+\|x\|_s) + \sqrtd{\Delta t} \,\| \C^{\frac{1}{2}}\xi^N \|_s.
\end{equs}
We have $\EE^{\pi^0} \Big[ \|\C^{\frac{1}{2}} \xi^N\|^p_s \Big] \leq \EE^{\pi^0} \Big[ \|\zeta\|^p_s \Big] < \infty$,
where $\zeta \dist \Normal(0,\C)$. 
Consequently, $\EE^{\pi^0} \Big[ \|\C^{\frac{1}{2}} \xi^N\|^p_s \Big]$ is uniformly bounded as a function of $N$, 
proving the lemma.
\end{proof}

Consider $y$ given by \eqref{eqn:proposal} and recall from \eqref{eqn:pmala-mala} that 
\be
y = \xmala + \rr^N(x,\delta).
\ee

\begin{lemma}\label{lem:appryx}
We have
\be
a^N(x,\delta) &\equiv \|y\|^2_{\C^N}  - \|\xmala\|^2_{\C^N}   \\
\EE^{\pi^N} a^N(x,\delta)  &\lesssim \delta^2
\ee
\end{lemma}

\begin{proof}
From \eqref{eqn:pmala-mala} we have
\be
\|y\|^2_{\C^N} - \|\xmala\|^2_{\C^N}  &= a^N(x,\delta) \\
a^N(x,\delta) &\equiv 2 \langle \xmala, \rr^N(x) \rangle_{\C^N} + \|\rr^N(x,\delta)\|^2_{\C^N}. \label{eqn:anrem}
\ee
From \eqref{eqn:remx}, we obtain
\be
|\langle\xmala, \rr^N(x,\delta) \rangle_{\C^N}| &=  |\langle \xmala, (\C^N)^{-1} \rr^N(x,\delta) \rangle| \\
&\leq \|\xmala\|_s \|\rr^N(x,\delta)\|_{\C^N} .
\ee
From Lemma \ref{lem:size:prop} we deduce that 
\be
\|\xmala\|_s \lesssim (1+ \delta) (1+\|x\|_s) + \sqrtd{\delta} \| \C^{\frac{1}{2}}\xi^N \|_s \,.
\ee
Combining this with Lemma \ref{lem:Remest} yields that 
\be
|\langle\xmala, \rr^N(x,\delta) \rangle_{\C^N}|  \lesssim \delta^2  (1 + \|x\|^2_s)(1 + \sqrt{\delta} \| \C^{\frac{1}{2}}\xi^N \|_s)\,.
\ee 
Thus
\be
\EE^{\pi^N}(|\langle\xmala, \rr^N(x,\delta) \rangle_{\C^N}|) \lesssim \delta^2 \,\EE^{\pi^N}(1 + \|x\|^2_s)(1 + \sqrt{\delta} \| \C^{\frac{1}{2}}\xi^N \|_s) \lesssim \delta^2. \label{eqn:anfinrem}
\ee
Thus from \eqref{eqn:anrem}, \eqref{eqn:anfinrem} and Lemma \ref{lem:Remest} we deduce that
\be
\EE^{\pi^N}(a^N(x,\delta) \lesssim \delta^2
\ee
and the proof is finished.
\end{proof}
%
%

\subsection{Gaussian approximation of $Q^N$}
\label{ssec:ga}

Recall the quantity $Q^N$ defined in Equation \eqref{eqn:QN}. This section proves that $Q^N$
has a Gaussian behavior in the sense that
\begin{equs} \label{e.gauss.app.error}
Q^N(x, \xi^N) = Z^N(x, \xi^N) \;+\; i^N(x,\xi^N) \;+\; \err^N(x,\xi^N)
\end{equs}
where the quantities $Z^N$ and $i^N$ are equal to
\begin{equs}
Z^N(x, \xi^N)&= -\frac{\ell^3}{4} - \frac{\ell^{\frac{3}{2}}}{\sqrt{2}} N^{-\frac{1}{2}} \sum_{j=1}^N \lambda_j^{-1} \xi_j x_j \label{e.Zn}\\
i^N(x,\xi^N) &= \frac{1}{2} (\ell \Delta t)^2 \Big( \|x\|^2_{\C^N} - \|(\C^N)^{\frac{1}{2}}\xi^N\|^2_{\C^N} \Big) \label{e.iN}
\end{equs}
with $i^N$ and $e^N$ small.
Thus the principal contributions to $Q^N$ 
comes from the random variable $Z^N(x,\xi^N)$.
Notice that, for each fixed $x \in \h^s$, the 
random variable $Z^N(x,\xi^N)$ is Gaussian.
Furthermore, the Karhunen-Lo\`eve expansion of 
$\pi_0$ shows that for $\pi_0$-almost every choice of 
function $x \in \h$ the sequence 
$\big\{ Z^N(x, \xi^N) \big\}_{N \geq 1}$ converges 
in law to the distribution of 
$Z_{\ell} \dist \Normal(-\frac{\ell^3}{4}, \frac{\ell^3}{2})$. 
The next lemma rigorously bounds the error terms 
$\err^N(x,\xi^N)$ and $i^N(x,\xi^N)$: we show that $i^N$ is an error term of order $\OO(N^{-\frac16})$ 
and $\err^N(x,\xi)$ is an error term of order $\OO(N^{-\frac13})$.
In Lemma \ref{lem:concentration} we then quantify the convergence
of $Z^N(x,\xi^N)$ to $Z_{\ell}.$

\begin{lemma} \label{lem:Gaussian_approx} 
{\bf (Gaussian Approximation)}
Let $p \geq 1$ be an integer.
Under Assumptions \ref{ass:1}, $Q^N(x, \xi^N)$ has the expansion given in 
\eqref{e.gauss.app.error} and the error terms
$i^N$ and $\err^N$ in the Gaussian approximation \eqref{e.gauss.app.error} satisfy
\begin{equation} \label{eqn:error_control}
\Big(\EE^{\pi^N} \big[|i^N(x,\xi^N)|^p \big]\Big)^{\frac{1}{p}} = \OO(N^{-\frac{1}{6}}) 
\qquad \text{and} \qquad
\Big(\EE^{\pi^N} \big[ |\err^N(x,\xi^N)|^p \big] \Big)^{\frac{1}{p}} = \OO(N^{-\frac{1}{3}}). 
\end{equation}
\end{lemma}
\begin{proof}
As in Lemma 4.4 of \cite{pill:stu:12},  without loss of generality, we suppose $p=2q$.
The quantity $Q^N$ is defined in Equation \eqref{eqn:QN} and
expanding terms leads to
\begin{equs}
Q^N(x,\xi^N) \;=\; I_1 \;+\; I_2 \;+\; I_3  + \; I_4
\end{equs}
where the quantities $I_1$, $I_2$, $I_3$ and $I_4$ are given by
\begin{equs}
I_1 &= -\frac{1}{2} \big( \|y\|_{\C^N}^2 - \|x\|_{\C^N}^2  \big)
- \frac{1}{4 \ell \Delta t} \big( \|x-y(1-\ell \Delta t)\|_{\C^N}^2 - \|y-x(1-\ell \Delta t)\|_{\C^N}^2 \big)  \\
I_2 &= -\Big(\Psi^N(y) - \Psi^N(x) \Big)
- \frac{1}{2} \Big( \bra{x-y(1-\ell \Delta t), \C^N \nabla \Psi^N(y)}_{\C^N} 
- \bra{y-x(1-\ell \Delta t), \C^N \nabla \Psi^N(x)}_{\C^N}\Big) \\
I_3 &= -\frac{1}{4\ell \Delta t}\Big\{ \| \ell \Delta t \, \C^N \nabla \Psi^N(y) + \rr^N(y,\delta) \|_{\C^N}^2 - \| \ell \Delta t \,\C^N \nabla \Psi^N(x) + \rr^N(x,\delta) \|_{\C^N}^2\Big\} \\
I_4&= -\frac{1}{2\ell \Delta t}\Big\{ \langle x - y(1- \ell \Delta t), \rr^N(y,\delta)\rangle_{\C^N} - 
\langle y - x(1- \ell \Delta t), \rr^N(x,\delta)\rangle_{\C^N} \Big\}.
\end{equs}

The term $I_1$ arises purely from the Gaussian part of the 
target measure $\pi^N$ and from the Gaussian part of
the proposal. The other terms come from the change of probability involving the functional $\Psi^N$.
By the calculation identical to page 2343 of \cite{pill:stu:12},  we can simplify the the term $I_1$ to be:
\begin{equs}
I_1 
&= -\frac{\ell \Delta t}{4} \Big( \|y\|_{\C^N}^2 - \|x\|_{\C^N}^2\Big). \label{eqn:I1main}
\end{equs}
The term $I_1$ is shown to be $\OO(1)$ and constitutes the main contribution to $Q^N$. Before analyzing $I_1$ 
in more detail, we show that $I_2$, $I_3$ and $I_4$ are $\OO(N^{-\frac{1}{3}})$:
\begin{equs} \label{e.I2.3.small}
\Big(\EE^{\pi^N}[I_2^{2q}]\Big)^{\frac{1}{2q}} + \Big(\EE^{\pi^N}[I_3^{2q}]\Big)^{\frac{1}{2q}} 
+ \Big(\EE^{\pi^N}[I_4^{2q}]\Big)^{\frac{1}{2q}} = \OO(N^{-\frac13}).
\end{equs}
\begin{itemize}
\item
By a calculation nearly identical to the one in Lemma 4.4 of \cite{pill:stu:12} (the only change being the use of our Lemma \ref{lem:size:prop} instead of their Lemma 4.2) we obtain that 
\begin{equs} \label{eqn:i2t}
\Big( \EE^{\pi^N}[I_2^{2q}] \Big)^{\frac{1}{2q}} \;=\;  \OO(N^{-\frac13}).
\end{equs}

\item
Using the definition of $\rr^N(x,\delta)$ from \eqref{eqn:remx}, we obtain that 
\be
\EE^{\pi^N}\big[ I_3^{2q} \big] &\lesssim \Delta t^{2q} \; \EE^{\pi^N}\Big[ |\bra{\nabla \Psi^N(x), \C^N \nabla \Psi^N(x)}|^q + |\bra{\nabla \Psi^N(y), \C^N \nabla \Psi^N(y)}|^q \Big]  \\
& \quad \quad +  \Delta t^{-2q} \; \EE^{\pi^N}\Big[ \| \rr^N(x,\delta) \|^{2q}_{\C^N} +  \| \rr^N(y,\delta) \|^{2q}_{\C^N} \Big].
\ee
Lemma \ref{lem:regularity} states $\C^N \nabla \Psi^N:\h^s \to \h^s$ is globally Lipschitz, 
with a Lipschitz constant that can be chosen uniformly in $N$. Therefore,
\begin{equs} \label{e.CPsi.bound}
\| \C^N \nabla \Psi^N(z) \|_s \lesssim 1 + \| z \|_s.
\end{equs}
Since $\|\C^N \nabla \Psi^N(z)\|_{\C^N}^2 = \bra{\nabla \Psi^N(z), \C^N \nabla \Psi^N(z)}$, the bound \eqref{eqn:psi2} gives
\begin{equs}
\EE^{\pi^N}\big[ I_3^{2q} \big]
&\lesssim \Delta t^{2q} \; \EE\Big[ \bra{\nabla \Psi^N(x), \C^N \nabla \Psi^N(x)}^q + \bra{\nabla \Psi^N(y), \C^N \nabla \Psi^N(y)}^q \Big] \\
&\lesssim  \Delta t^{2q} \; \EE^{\pi^N}\Big[ (1+\|x\|_s)^{2q} + (1+\|y\|_s)^{2q} \Big] \\
&\lesssim  \Delta t^{2q} \; \EE^{\pi^N}\Big[ 1+\|x\|^{2q}_s + \|y\|^{2q}_s \Big] \;\lesssim\;  \Delta t^{2q}
\;=\; \Big(N^{-\frac13}\Big)^{2q}. \label {eqn:i3term1}
\end{equs}
Similarly, from Lemma \ref{lem:Remest} and \ref{lem:size:prop},
\be
\Delta t^{-2q} \; \EE^{\pi^N}\Big[ \| \rr^N(x,\delta) \|^{2q}_{\C^N} +  \| \rr^N(y,\delta) \|^{2q}_{\C^N} \Big]
&\lesssim \Delta t^{6q} \,\EE^{\pi^N}\Big[ 1+\|x\|^{2q}_s + \|y\|^{2q}_s \Big] \;\lesssim\;  \Delta t^{6q} \\
&\;\lesssim\; \Big(N^{-\frac13}\Big)^{6q} \lesssim  \Big(N^{-\frac13}\Big)^{2q}. \label{eqn:i3term2}
\ee
Thus from \eqref{eqn:i3term1} and \eqref{eqn:i3term2}, we conclude that 
\be
\Big(\EE^{\pi^N}[I_3^{2q}]\Big)^{\frac{1}{2q}} \lesssim \Big(N^{-\frac13}\Big)^{2q}. \label{eqn:i3t}
\ee
\item 
Finally, we tackle the term $I_4$:
\be
\EE^{\pi^N}\big[ I_4^{2q} \big]  &\lesssim \Delta t^{-2q} \; \EE\Big[ \|x - y(1- \ell \Delta t)\|^{2q}_s \,\| (\C^N)^{-1} \rr^N(y,\delta)\|^{2q}_s  \\
& \quad \quad \quad \quad \quad  +  \|y - x(1- \ell \Delta t)\|^{2q}_s \,\| (\C^N)^{-1} \rr^N(x,\delta)\|^{2q}_s\Big].
\ee
From Lemma \ref{lem:size:prop}, we obtain that $\EE^{\pi^N}( \|y - x(1- \ell \Delta t)\|^{4q}_s) \lesssim (\Delta t)^{2q} \cdot (1 + \|x\|^{4q}_s) $ and $\EE^{\pi^N}( \|x - y(1- \ell \Delta t)\|^{2q}_s) \lesssim (\Delta t)^{2q} \cdot (1 + \|x\|^{4q}_s). $ 
Similarly, from Lemma \ref{lem:Remest} we gather that $\EE^{\pi^N}\|\rr^N(x,\delta)\|^{4q}_{\C^N} \lesssim \delta^{8q}  (1 + \|x\|^{4q}_s)$. Putting these two together and using the Cauchy-Schwartz inequality gives,
\be
 \EE^{\pi^N}\big[ I_4^{2q} \big] &\lesssim \Delta t^{3q} \, \EE^{\pi^N}\Big[ 1+\|x\|^{2q}_s + \|y\|^{2q}_s \Big] \lesssim  \Big(N^{-\frac13}\Big)^{2q}. \label{eqn:i4t}
\ee
\end{itemize}
Equations \eqref{eqn:i2t}, \eqref{eqn:i3t} and \eqref{eqn:i4t} imply the requisite estimate in \eqref{e.I2.3.small}. 

\noindent Next, we tackle the term $I_1$. Recall from from \eqref{eqn:I1main} that
\be
I_1
&= -\frac{\ell \Delta t}{4} \Big( \|y\|_{\C^N}^2 - \|x\|_{\C^N}^2\Big).
\ee
From Lemma \ref{lem:appryx} we obtain that
\be \label{eqn:gapp1}
\|y\|_{\C^N}^2 = \|\xmala\|^2_{\C^N} + a^N(x,\ell \Delta t), \quad \quad \EE^{\pi^N} a^N(x,\ell \Delta t)  \lesssim (\Delta t)^2.
\ee
Consequently, 
\be
I_1
&= -\frac{\ell \Delta t}{4} \Big( \|\xmala\|_{\C^N}^2 - \|x\|_{\C^N}^2\Big)  -\frac{\ell \Delta t}{4} a^N(x,\ell \Delta t).
\ee
From Lemma 4.4 of \cite{pill:stu:12}, we deduce that
\be \label{eqn:gapp2}
 -\frac{\ell \Delta t}{4} \Big( \|\xmala\|_{\C^N}^2 - \|x\|_{\C^N}^2\Big) 
&= Z^N(x,\xi^N) \;+\; i^N(x,\xi^N) \;+ \; b^N(x,\xi^N)
\end{equs}
with $Z^N(x,\xi^N)$ and $i^N(x,\xi^N)$ given by Equation \eqref{e.Zn} and \eqref{e.iN} and
\begin{equs}
\Big( \EE^{\pi^N} \big[ b^N(x,\xi^N)^{2q} \big] \Big)^{\frac{1}{2q}} \;=\; \OO(N^{-\frac13}).
\end{equs}
Lemma 4.4 of \cite{pill:stu:12} also shows that
\begin{equs} \label{e.in.bound}
\Big( \EE^{\pi^N} \big[ i^N(x,\xi^N)^{2q} \big] \Big)^{\frac{1}{2q}} = \OO(N^{-\frac16}).
\end{equs}
The proof of the lemma now follows from \eqref{e.I2.3.small}, \eqref{eqn:gapp1} and \eqref{eqn:gapp2}.
\end{proof}

\noindent

We recall Lemma 4.5 of \cite{pill:stu:12}:
\begin{lemma} \cite{pill:stu:12}(Lemma 4.5) \label{lem:concentration} {\bf (Asymptotic independence)}
Let $p \geq 1$ be a positive integer and $f:\RR \to \RR$ be a $1$-Lipschitz function. 
Consider error terms $\err^N_{\star}(x,\xi)$ satisfying 
\begin{equs}
\lim_{N \to \infty} \; \EE^{\pi^N}[\err^N_{\star}(x,\xi^N)^p] = 0. 
\end{equs}
Define the functions $\bar{f}^N: \RR \to \RR$ and the constant $\bar{f} \in \RR$ by
\begin{equs}
\bar{f}^N(x) = \EE_x\Big[ f\big(Z^N(x,\xi^N) + \err^N_{\star}(x,\xi^N) \big) \Big]
\qquad \qquad \text{and} \qquad \qquad
\bar{f} = \EE[f(Z_{\ell})].
\end{equs}
Then the function $f^N$ is highly concentrated around its mean in the sense that
\begin{equs}
\lim_{N \to \infty} \; \EE^{\pi^N}\Big[ |\bar{f}^N(x) - \bar{f}|^p \Big] = 0.
\end{equs}
\end{lemma}

\begin{corollary} \label{rem:concentration}
Let $p \geq 1$ be a positive.
The local mean acceptance probability $\alpha^N(x)$ defined in Equation \eqref{e.loc.accept} satisfies
\begin{equs}
\lim_{N \to \infty} \; \EE^{\pi^N}\big[ |\alpha^N(x) - \alpha(\ell)|^p \big] = 0.
\end{equs}
\end{corollary}
\begin{proof}
The function $f(z)=1 \wedge e^z$ is $1$-Lipschitz and $\alpha(\ell) = \EE[f(Z_{\ell})]$. Also, 
\begin{equs}
\alpha^N(x) = \EE_x\Big[ f(Q^N(x,\xi^N)) \Big] = \EE_x \big[f(Z^N(x,\xi^N) + \err^N_{\star}(x,\xi^N) \Big]
\end{equs}
with $\err^N_{\star}(x,\xi^N) = i^N(x,\xi^N) +\err^N(x,\xi^N)$. Lemma \ref{lem:Gaussian_approx} shows that 
$$\lim_{N \to \infty} \; \EE^{\pi^N}[\err^N_{\star}(x,\xi)^p] = 0$$ and therefore
Lemma \ref{lem:concentration} gives the conclusion.
\end{proof}

\subsection{Drift approximation}
\label{ssec:da}

This section proves that the approximate drift function $d^N:\h^s \to \h^s$ defined in Equation \eqref{eq:drift}
converges to the drift function $\mu:\h^s \to \h^s$ of the limiting diffusion \eqref{eqn:spdelim}.

\begin{lemma} {\bf (Drift Approximation)} \label{lem.drift.approx}
Let Assumptions \ref{ass:1} hold. 
The drift function $d^N:\h^s \to \h^s$ converges to $\mu$ in the sense that 
\begin{equs}
\lim_{N \to \infty} \, \EE^{\pi^N}\Big[ \|d^N(x)-\mu(x)\|^2_s \Big] = 0. 
\end{equs}
\end{lemma}
\begin{proof}
Now that we have established the relevant estimates, the proof of this lemma is nearly identical to that of Lemma 4.7 of \cite{pill:stu:12}, but also needs to account for the extra error term induced by the proximal operator. The approximate drift $d^N$ is given by Equation \eqref{eq:drift}.
The definition of the local mean acceptance probability $\alpha^N(x)$
given by Equation \eqref{e.loc.accept} shows that $d^N$
can also be expressed as
\begin{equs} \label{eqn:driftrem}
d^N(x) = \Big( \alpha^N(x) \alpha(\ell)^{-1} \Big) \mu^N(x) + \mathcal{R}_{\mathrm{Prox}}^N (x, \Delta t) 
+ \sqrt{2 \ell} h(\ell)^{-1} (\Delta t)^{-\frac12} \eps^N(x)
\end{equs}
where $\mu^N(x) = -\Big( P^N x + \C^N \nabla \Psi^N(x)\Big)$; the term $\eps^N(x)$ is defined by
\begin{equs}
\eps^N(x)
\;=\; \EE_x\big[ \gamma^N(x,\xi^N) \, \C^{\frac{1}{2}} \xi^N \big] 
\;=\; \EE_x \big[ \big( 1 \wedge e^{Q^N(x,\xi^N)} \big) \, \C^{\frac{1}{2}} \xi^N \big]
\end{equs}
and the term $\mathcal{R}_{\mathrm{Prox}}^N (x, \Delta t)$ is the error term induced by the proximal approximation:
\be
\mathcal{R}^N_{\mathrm{Prox}}(x,\Delta t) = {\alpha^N(x) \over h(\ell)} {\rr^n(x,\ell \Delta) \over \Delta t}.
\ee
To prove Lemma \ref{lem.drift.approx} it suffices to verify that
\begin{equs}
&\lim_{N \to \infty} 
\EE^{\pi^N} \Big[ \big\| \big( \alpha^N(x) \alpha(\ell)^{-1} \big) 
\mu^N(x) - \mu(x) \big\|_s^2 \Big]= 0 \label{e.drift.cond1}\\
&\lim_{N \to \infty}  \EE^{\pi^N}\|\mathcal{R}^N_{\mathrm{Prox}}(x,\Delta t)\|^2_s = 0 \label{e.drift.cond2}\\
&\lim_{N \to \infty} (\Delta t)^{-1} \; 
\EE^{\pi^N} \Big[ \|\eps^N(x) \|_s^2 \Big] = 0. \label{e.drift.cond3}
\end{equs}
\begin{itemize}
\item
Equation \eqref{e.drift.cond1}  follows directly from Lemma 4.7 of \cite{pill:stu:12}.
\item Next, using the fact that $|\alpha^N(x)| \leq 1$ and Lemma \ref{lem:Remest},
\be
\|\mathcal{R}^N_{\mathrm{Prox}}(x,\Delta t)\|^2_s &\lesssim \Big( {\alpha^N(x)} \Big)^2   \Big \|{\rr^n(x,\ell \Delta t) \over \Delta t}\Big\|^2_s \\
&\lesssim (\Delta t)^2 (1 + \|x\|^2_s)
\ee
and thus we have
\be
 \lim_{N \to \infty} \EE^{\pi^N}\|\mathcal{R}^N_{\mathrm{Prox}}(x,\Delta t)\|^2_s = \lim_{N \to \infty} N^{-2/3} \EE^{\pi^N}(1 + \|x\|^2_s) = 0
\ee
establishing \eqref{e.drift.cond2}.
\item
Let us prove Equation \eqref{e.drift.cond3}. 
If the Bernoulli random variable $\gamma^N(x,\xi^N)$ were independent from the noise term $(\C^N)^{\frac{1}{2}} \xi^N$, 
it would follow that $\eps^N(x)=0$. In general $\gamma^N(x,\xi^N)$ is not independent from $(\C^N)^{\frac12} \xi^N$
so that $\eps^N(x)$ is not equal to zero. Nevertheless, as 
quantified by Lemma \ref{lem:concentration}, 
the Bernoulli random variable $\gamma^N(x,\xi^N)$ is asymptotically 
independent from the current position $x$ and from the noise term $(\C^N)^{\frac12} \xi^N$. 
Consequently, we can prove in Equation \eqref{e.eps.small} that the quantity $\eps^N(x)$ is small. 
To this end, we establish that each component $\bra{\eps(x), \pphi_j}^2_s$ satisfies
\begin{equs} \label{e.component.eps}
\EE^{\pi^N} \big[ \bra{\eps^N(x), \pphi_j}^2_s \big] \quad \lesssim \quad N^{-1} \EE^{\pi^N}[\bra{x,\pphi_j}^2_s] + N^{-\frac{2}{3}} (j^s \lambda_j)^2.
\end{equs}
Summation of Equation \eqref{e.component.eps}
over $j=1, \ldots,N$ leads to
\begin{equs}
\EE^{\pi^N} \; \Big[ \|\eps^N(x)\|_s^2 \Big] 
\quad \lesssim \quad N^{-1} \EE^{\pi^N}\big[\|x\|_s^2 \big] + N^{-\frac{2}{3}} \, \tr_{\h^s}(\C_s) 
\quad \lesssim \quad N^{-\frac{2}{3}},\\  \label{e.eps.small}
\end{equs}
which gives the proof of Equation \eqref{e.drift.cond3}. 
To prove Equation \eqref{e.component.eps} for a fixed index $j \in \mathbb{N}$, the quantity 
$Q^N(x,\xi)$ is decomposed as a sum of a term independent from $\xi_j$ and another remaining term of small magnitude.
To this end we introduce
\begin{equs} \label{e.QQi}
\left\{
\begin{array}{ll}
Q^N(x,\xi^N) &= Q^N_j(x,\xi^N) + Q^N_{j,\perp}(x,\xi^N) \\
Q^N_{j}(x,\xi^N) &= -\frac{1}{\sqrt{2}} \ell^{\frac32} N^{-\frac12} \lambda_j^{-1} x_j \xi_j 
- \frac{1}{2} \ell^2 N^{-\frac23} \lambda_j^2 \xi_j^2
+\err^N(x,\xi^N).
\end{array}
\right.
\end{equs}
The definitions of $Z^N(x,\xi^N)$ and $i^N(x,\xi^N)$ in Equation \eqref{e.Zn} and \eqref{e.iN} 
readily show that $Q^N_{j,\perp}(x,\xi^N)$ is independent from $\xi_j$.
The noise term satisfies $\C^{\frac12} \xi^N = \sum_{j=1}^N (j^s \lambda_j) \xi_j \pphi_j$.
Since $Q^N_{j,\perp}(x,\xi^N)$ and $\xi_j$ are independent and $z \mapsto 1 \wedge e^z$ is $1$-Lipschitz, it follows that
\begin{equs}
\bra{\eps^N(x), \pphi_j}^2_s
&= (j^s \lambda_j)^2 \; \Big( \EE_x\big[ \big(1 \wedge e^{Q^N(x,\xi^N)} \big) \; \xi_j \big] \Big)^2 \\
&= (j^s \lambda_j)^2 \; \Big( \EE_x\big[ [\big(1 \wedge e^{Q^N(x,\xi^N)} \big)-\big(1 \wedge e^{Q^N_{j,\perp}(x,\xi^N)} \big)] \; \xi_j \big] \Big)^2 \\
&\lesssim (j^s \lambda_j)^2 \EE_x \big[ |Q^N(x,\xi^N))-Q^N_{j,\perp}(x,\xi^N)|^2\big]\\
&= (j^s \lambda_j)^2 \EE_x \big[ Q^N_{j}(x,\xi^N)^2 \big].
\end{equs}
By Lemma \ref{lem:Gaussian_approx} $\EE^{\pi^N}\big[ \err^N(x,\xi^N)^2 \big] \lesssim N^{-\frac{2}{3}}$. Therefore, 
\begin{equs}
 (j^s \lambda_j)^2 \EE^{\pi^N} \big[ Q^N_{j}(x,\xi^N)^2\big]
&\lesssim  (j^s \lambda_j)^2 \Big\{ N^{-1} \lambda_j^{-2} \EE^{\pi^N}\big[ x^2_j \xi^2_j \big]
+ N^{-\frac43} \EE^{\pi^N}\big[ \lambda_j^4 \xi_j^4 \big]
+\EE^{\pi^N}\big[\err^N(x,\xi)^2 \big] \Big\} \\
&\lesssim N^{-1} \; \EE^{\pi^N}\big[ (j^{s}x_j)^2 \xi^2_j \big] + (j^s \lambda_j)^2 (N^{-\frac{4}{3}} + N^{-\frac{2}{3}}) \\
&\lesssim N^{-1} \; \EE^{\pi^N}\big[ \bra{x,\pphi_j}_s^2 \big] + (j^s \lambda_j)^2 N^{-\frac{2}{3}} \\
&\lesssim N^{-1} \; \EE^{\pi^N}\big[ \bra{x,\pphi_j}_s^2 \big] + (j^s \lambda_j)^2 N^{-\frac{2}{3}},
\end{equs}
which finishes the proof of Equation \eqref{e.component.eps}.
\end{itemize}
Thus we have established \eqref{e.drift.cond1}, \eqref{e.drift.cond2} and \eqref{e.drift.cond3} and the proof is finished.
\end{proof}

\subsection{Noise approximation}
\label{ssec:na}

Recall the definition \eqref{eq:mart} of the martingale
difference $\Gamma^{k,N}.$ In this section we estimate the 
error in the approximation $\Gamma^{k,N} \approx \Normal(0,\C_s)$.
To this end we introduce the covariance operator 
\begin{equs}
D^N(x) = \EE_x \Big[ \Gamma^{k,N} \otimes_{\h^s} \Gamma^{k,N} \;| x^{k,N}=x \Big].
\end{equs}
For any $x, u,v \in \h^s$ the operator $D^N(x)$ satisfies
\begin{equs}
\EE\Big[\bra{\Gamma^{k,N},u}_s \bra{\Gamma^{k,N},v}_s \;|x^{k,N}=x \Big]  \;=\; \bra{u, D^N(x) v}_s.
\end{equs}
The next lemma gives a quantitative version of the approximation $D^N(x) \approx \C_s$.

\begin{lemma} \label{lem:diffus} Let Assumptions \ref{ass:1} hold. 
For any pair of indices $i,j \geq 0$ the operator $D^N(x):\h^s \to \h^s$ satisfies
\begin{equation}
\lim_{N \to \infty} \quad \EE^{\pi^N} \big| \bra{\pphi_i, D^N(x) \pphi_j}_s - \bra{\pphi_i, \C_s \pphi_j}_s \big|   
\;=\; 0 \label{e.noise.approx.1}
\end{equation}
and, furthermore,
\begin{equation}
\lim_{N \to \infty} \quad \EE^{\pi^N} \big| \tr_{\h^s}(D^N(x)) - \tr_{\h^s}(\C_s) \big| 
\;=\; 0. \label{e.noise.approx.2}
\end{equation}
\end{lemma}
\begin{proof}
This lemma follows directly from Lemma 4.8 of \cite{pill:stu:12}, since the only estimate 
needed for the proof of Lemma 4.8 of \cite{pill:stu:12} is the Gaussian approximation and the estimate for $\err^N(x,\xi^N)$ established in Lemma \ref{lem:Gaussian_approx}.
Thus the proof is finished.
\end{proof}

\subsection{Martingale Invariance Principle} \label{sec:bweakconv}
This section proves that the process $W^N$ defined in 
Equation \eqref{e.WN} converges to a Brownian motion.
\begin{proposition} \label{lem:bweakconv}%
Let Assumptions \ref{ass:1} hold.
Let $z^0 \sim \pi$ and $W^N(t)$ the process defined in equation \eqref{e.WN} and
$x^{0,N} \dist \pi^N$ the starting position of the Markov chain $x^N$. Then
\begin{equation}
 (x^{0,N}, W^N) \longweak (z^0,W),
\end{equation}
where $\longweak$ denotes weak convergence in $\h^s \times C([0,T];\h^s)$, and $W$ is
a $\h^s$-valued Brownian motion with covariance operator $\C_s$.
Furthermore the limiting Brownian motion $W$ is independent of the initial condition $z^0$.
\end{proposition}

\begin{proof}

This proof involves verifying three conditions of Proposition $5.1$ of \cite{Berg:86} and is identical to that of Proposition 4.10 of \cite{pill:stu:12}. The only change required is to use our  Lemma \ref{lem:diffus} instead of their Lemma 4.8. Therefore we omit the details of the rest of the proof. 
\end{proof}

\section{Closing Comments} \label{sec:close}
There are a number of related issues that are of great practical interest: 
\begin{itemize}
\item In Theorem 1 of \cite{cruc:23}, the authors extended Theorem \ref{thm:pmalgaus} to a general class of product measures.
\item As mentioned in the introduction, we choose $\lambda  = \delta$. In \cite{cruc:23}, it is shown that for differentiable targets, this choice is optimal. When $\delta \rightarrow 0$ quicker than $\lambda$, proximal MALA is less efficient than MALA.
\item Of course, the most interesting case is when the log-target is not differentiable. In this scenario, \cite{cruc:23} show that the applicability of proximal MALA comes at a cost -- the algorithm scales smaller than $O(N^{-\frac{1}{3}})$ and is less efficient than its smooth counterpart.
\item A similar result should be of interest when proximal functions are used for implementing the Hybrid Monte Carlo algorithm; see \cite{chaa:16}.
\end{itemize}
\section*{Acknowledgements} 
The author thanks Gareth Roberts and Andrew Stuart for introducing him to optimal scaling of Markov Chain Monte Carlo Methods and several inspiring conversations,
Dean Foster and Robert Stine for introducing him to proximal algorithms and related ADMM methods, Alex Belloni, David Dunson, Lucas Janson, Eric Laber, Jonathan Mattingly, Sayan Mukherjee, Debdeep Pati, Christian Robert, Aaron Smith, Pragya Sur and Alex Thiery for very helpful comments and their collaboration over the years. Part of this work was done when the author was an Amazon Scholar. He thanks Robert Stine and Andrew Ruud for their encouragement to work on this problem.

\vskip 0.2in
\bibliography{jmlr_final_ref}

\end{document}